\documentclass[10pt,aps,prd,twocolumn,showpacs,amsmath,amssymb,nofootinbib,eqsecnum,preprintnumbers,superscriptaddress]{revtex4-2}
\usepackage{amsmath,amssymb}
\usepackage{enumitem}  
\usepackage[usenames, dvipsnames]{color} 
\usepackage{graphicx} 
\usepackage{comment}
\usepackage[normalem]{ulem}
\usepackage[utf8]{inputenc}
\UseRawInputEncoding 
\usepackage{url}
 \usepackage{xcolor}
 \usepackage{multirow}
 \usepackage{amssymb}
\usepackage{pifont}
\usepackage{amsthm}
\newcommand{\xmark}{\ding{55}}%

\newcommand{\C}{ {\mbox{\tiny (C)}} }
\newcommand{\PD}{ {\mbox{\tiny (PD)}} }

\usepackage{hyperref}  

\usepackage{graphicx}
\usepackage{dcolumn}
\usepackage{bm}

\newcommand{\be}{\begin{equation}}
\newcommand{\ee}{\end{equation}}
\newcommand{\ba}{\begin{eqnarray}}
\newcommand{\ea}{\end{eqnarray}}

\newcommand{\beq}{\begin{equation}}
\newcommand{\eeq}{\end{equation}}
\newcommand{\beqa}{\begin{eqnarray}}
\newcommand{\eeqa}{\end{eqnarray}}

\newcommand{\im}{\mathrm{i}}
\newcommand{\dd}{{\rm{d}}}
\newtheorem{statement}{Lemma}

\begin{document}


\title{Hidden symmetries and separability structures of Ovcharenko--Podolsk{\'y} and conformal-to-Carter spacetimes}%

\author{Finnian Gray}

\email{finnian.gray@univie.ac.at}
\affiliation{University of Vienna, Faculty of Physics, Boltzmanngasse~5, 1090 Vienna, Austria}

\author{David Kubiz\v n\'ak}

\email{david.kubiznak@matfyz.cuni.cz}

\affiliation{Institute of Theoretical Physics, Faculty of Mathematics and Physics,
Charles University, Prague, V Hole{\v s}ovi{\v c}k{\' a}ch 2, 180 00 Prague 8, Czech Republic}

\author{Hryhorii Ovcharenko}

\email{hryhorii.ovcharenko@matfyz.cuni.cz}

\affiliation{Institute of Theoretical Physics, Faculty of Mathematics and Physics,
Charles University, Prague, V Hole{\v s}ovi{\v c}k{\' a}ch 2, 180 00 Prague 8, Czech Republic}

\author{Ji{\v r}{\'i} Podolsk{\' y}}

\email{jiri.Podolsk\'y@matfyz.cuni.cz}

\affiliation{Institute of Theoretical Physics, Faculty of Mathematics and Physics,
Charles University, Prague, V Hole{\v s}ovi{\v c}k{\' a}ch 2, 180 00 Prague 8, Czech Republic}

\date{November 26, 2025}

\begin{abstract}
Recently, a remarkable new class of spacetimes describing black holes immersed in a non-aligned electromagnetic field has been found. 
While still of type D, this class goes beyond the famous Pleba\'nski--Demia\'nski family. Here we demonstrate that the whole class admits a hidden symmetry encoded in the non-degenerate conformal Killing--Yano 2-form. Interestingly, as a direct consequence of non-alignment of the electromagnetic field and contrary to the Pleba\'nski--Demia\'nski class (where the field is aligned), such a symmetry no longer generates the full ``tower of symmetries". 
Despite this, it enables one to separate variables in massless Hamilton--Jacobi, conformal wave, and massless Dirac equations, as well as allows one to tackle massless vector and tensor perturbations. These results provide a useful mathematical tool for discussing numerous (astrophysical) applications to be described by these metrics. Moreover, as we shall show, the novel spacetimes provide an interesting test-ground for studying the recently defined Penrose charges whose existence is intrinsically related to hidden rather than explicit symmetries.
\end{abstract}

 \maketitle

\section{Introduction}

Killing tensors and Killing--Yano forms are noteworthy {\em hidden symmetries} of black hole spacetimes in general relativity. 
Discovered first by studying the geodesic motion \cite{Carter:1968rr}, they underlie many of the fundamental properties of rotating Kerr geometry \cite{Kerr:1963ud}, and more generally of {\em Kerr--NUT--AdS spacetimes} in four \cite{carter1968new, Carter:1968ks} and higher \cite{Chen:2006xh} dimensions. 
In particular, they stand behind the separability of Hamilton--Jacobi, Klein--Gordon, Dirac, and (massive) vector equations in these spacetimes, determine the special algebraic type of the solution, are related to its associated Kerr--Schild form, and lead {\em uniquely} to the (off-shell) Kerr--NUT--AdS geometry,%
\footnote{Perhaps the `cleanest'' (and well motivated) derivation of the Kerr spacetime is to assume the existence of a non-degenerate Killing--Yano 2-form. 
This fixes the geometry to the off-shell Kerr--NUT--AdS form. 
By further imposing the vacuum Einstein equations together with demanding regularity on the axes then uniquely leads to the Kerr geometry.
Similar construction remains valid also in higher dimensions  \cite{Krtous:2008tb}.} 
see the review \cite{Frolov:2017kze} and references therein for more details.
In this paper we shall primarily focus on four dimensions, and call the corresponding Kerr--NUT--AdS spacetimes (potentially also  with an aligned electromagnetic field) in honour of their discoverer the {\em Carter spacetimes.}

Interestingly, in four dimensions there exists even a larger family of black hole spacetimes, encoded in the so called {\em Pleba\'nski--Demia\'nski (PD) class} \cite{Plebanski:1976gy}. This is the most general type D solution of Einstein--Maxwell equations with an {\em aligned} electromagnetic field.  
Apart from rotating charged black holes, possibly endowed with a cosmological constant and NUT parameter, there is an additional parameter allowing black holes to uniformly accelerate, being pulled by unbalanced cosmic strings.
Surprisingly, the corresponding metric is conformally related to the off-shell Carter  spacetime, inheriting some of its  hidden symmetries. 
Namely, a {\em weaker structure} of the so-called conformal Killing--Yano forms and conformal Killing tensors is now present, e.g. \cite{Kubiznak:2007kh}, allowing for  integrability of null geodesics, and separability of conformal test field equations in these spacetimes. 

Very recently, a novel remarkable class of type D spacetimes with {\em non-aligned} electromagnetic fields was discovered by {\em Ovcharenko and Podolsk\'y (OP)} \cite{Ovcharenko:2025cpm}.  
While the full class still remains to be explored and interpreted, it gives rise, as a particular special case, to a new rotating black hole spacetime emersed in the Bertotti--Robinson magnetic universe---the so called {\em Kerr--BR spacetime}  \cite{Podolsky:2025tle}. Such a solution, which describes a  rotating black hole in a magnetic environment, is of great potential astrophysical interest and has been a subject of very active recent investigations \cite{Zeng:2025olq, Wang:2025vsx, Zeng:2025tji, Wang:2025bjf,Astorino:2025lih, Ali:2025beh, Vachher:2025jsq, Zhang:2025ole, Rueda:2025lgq, Liu:2025wwq, Ahmed:2025ril, Ortaggio:2025sip}.

The aim of this paper is to further investigate the {\em properties} of the new OP spacetime family, focusing mainly on its mathematical structure, hidden symmetries, and integrability properties. 
First, we shall show that (perhaps not surprisingly and  similarly to the PD case) the spacetime admits a hidden symmetry of the conformal Killing--Yano 2-form, and associated with it conformal Killing tensors. 
Contrary to the PD spacetime, however, in the presence of non-aligned electromagnetic field there is no clear path to derive from this object the existence of explicit symmetries. 
This observation motivates us to investigate even more general {\em conformal-to-Carter} geometries, and to probe under which conditions their hidden symmetry (inherited from  Carter's geometry) implies the existence of explicit symmetries. 

Second, we shall investigate the separability of various physical test field equations in the OP spacetimes. In particular, we will demonstrate the separability of the massless Hamilton--Jacobi, conformal wave, massless Dirac, and massless vector and tensor perturbations in these geometries. 
Some of these results can be directly traced back to the existence of hidden symmetries and from them derived associated symmetry operators. 

Finally, we shall investigate some special sub-cases of these general results, focusing mainly on the non-twisting and Kerr--BR spacetimes. 
The latter will be shown to provide a first ever example of a spacetime with {\em conserved charges derived completely from the hidden symmetry} of conformal Killing--Yano tensors, as recently proposed in \cite{Hull:2025ivk}, {where the corresponding current (see \eqref{Current1form} below) is no longer generated from a Killing vector field}.

The paper is structured as follows. 
In Sec.~\ref{Sec:HS} we recall the basic definitions and properties of various hidden symmetries. 
The two  well known type D classes of spacetimes with aligned electromagnetic fields, namely the Carter and PD spacetimes, are reviewed---together with their hidden symmetries---in Sec.~\ref{Sec:Carter&PD}. 
The novel OP spacetime is introduced in Sec.~\ref{Sec:OP}, and its symmetries are then discussed. 
Even more general conformal-to-Carter spacetimes, are studied in Sec.~\ref{Sec:Tower}, with a particular emphasis on the derivation of explicit symmetries from its hidden symmetry. 
Various separability structures are studied in Sec.~\ref{Sec:Separability}. 
The non-twisting case of the OP spacetime is analyzed in Sec.~\ref{Sec:Nontwisting}, while Sec.~\ref{Sec:SpecialCases} discusses some other interesting special cases, including the Kerr--BR and Schwarzschild--BR metrics, and their Penrose charges. We conclude in Sec.~\ref{Sec:Conclusions}. In Appendix~\ref{App:onshell}, we gather the explicit form of the metric functions for the on-shell Carter, PD, and OP spacetimes.   Appendix~\ref{App:B} contains the coordinate transformations which link various forms of the geometrically preferred metrics discussed in the main text.

\section{Hidden symmetries}\label{Sec:HS}

In this section, we shall briefly overview the basic properties of objects responsible for hidden symmetries in general relativity. 
We start with {\em conformal Killing--Yano} (CKY) tensors and their special cases. 
A CKY tensor is a $p$-form $h$ whose covariant derivative has a vanishing harmonic part, that is, it can be split into the {\em exterior derivative} and {\em divergence} parts as follows  \cite{Kashiwada, Tachibana}: 
\begin{equation}\label{CKY}
\nabla_{\!\mu}h_{\alpha_1\dots \alpha_p}=
\nabla_{\![\mu}h_{\alpha_1\dots \alpha_p]}+
\frac{p}{D-p+1}\,g_{\mu[\alpha_1\!}\nabla_{\!|\kappa|}h^{\kappa}_{\ \alpha_2\dots\alpha_p]}\,,
\end{equation}
where $D$ stands for the number of spacetime dimensions.
The word {\em conformal} refers to the fact that such objects behave ``nicely'' under conformal transformations. 
Namely, considering the Weyl transformation 
\be\label{Weyl} 
g\to g'=\frac{1}{\Omega^2}\, g\,,
\ee 
the appropriately rescaled $p$-form  $h'$, obtained as, 
\be\label{conformalProp} 
h \to h' =\frac{1}{\Omega^{p+1}}\,h\,,
\ee 
remains a CKY tensor of the metric $g'$. 
It is largely this property, that is responsible for the results discussed in this paper.

Another useful property is that the CKY equation \eqref{CKY} is invariant under the {\em Hodge duality}. 
Namely, the exterior derivative part transforms into the divergence part and vice versa. 
This in particular implies that the dual $\star h$ is a CKY tensor whenever $h$ is. 
However, one can do even more.
There are two special subclasses of CKY tensors of particular interest:
i) {\em Killing--Yano} (KY) tensors \cite{Yano} are CKYs with zero divergence part in \eqref{CKY}, and ii) {\em closed conformal Killing--Yano} (CCKY) tensors are CKYs with vanishing exterior derivative part in \eqref{CKY}. 
As obvious from the above discussion, these subclasses transform one into another under Hodge duality.

From now on, let us set $D=4$ and concentrate on the case of CKY 2-forms, relevant for black hole spacetimes in four dimensions. 
The above defining equation \eqref{CKY} then significantly simplifies and reads 
\begin{equation}\label{CKY4D}
\nabla_{\!\mu}h_{\alpha\beta}=
\nabla_{\![\mu}h_{\alpha\beta]}+
{2}\, g_{\mu[\alpha} \xi_{\beta]}\;,
\end{equation}
where we have defined the vector ${\xi}$ as the normalized divergence of $h$, namely:
\begin{equation}\label{xi}
\xi_{\alpha}=\frac{1}{3}\nabla_{\!\kappa}\,h^{\kappa}_{\ \alpha}\;.
\end{equation}
It can be shown, e.g. \cite{Jezierski_2006}, that such a vector obeys the following equation:
\be \label{xiKVeq}
\nabla_{(\alpha}\,\xi_{\beta)}=\frac{1}{2}h_{(\alpha|\gamma|}R^\gamma{}_{\beta)}\,,
\ee 
where $R_{\alpha\beta}$ is the Ricci tensor. 
Thus, whenever the spacetime is an Einstein space, such $\xi$ is necessarily a Killing vector. However, in more general spacetimes this need not be the case. 
We shall return to this question in connection with conformal-to-Carter metrics in Sec.~\ref{Sec:Tower}.
Similarly, one can also define another ``dual'' vector $\zeta$, corresponding to the dual CKY 2-form $k=\star h$, namely 
\begin{equation}\label{zeta}
\zeta_{\alpha}=\frac{1}{3}\nabla_{\!\kappa}\,k^{\kappa}_{\ \alpha}\;.
\end{equation}
If $h$ is a KY, CCKY tensor then $\xi$, $\zeta$ vanishes, respectively.

As shown recently in \cite{Hull:2025ivk} (see also \cite{Lindstrom:2021dpm,Lindstrom:2022qjx}, CKY forms give rise to conserved higher-form charges. 
In particular, having a CKY 2-form \eqref{CKY4D} in 4D, we can construct a ``{\em primary current}', given by 
\be 
J_{\mu\nu}^{h}=-\frac{1}{2}\Bigl(R_{\mu\nu \alpha\beta}\,h^{\alpha\beta}+4R^\alpha{}_{[\mu}h_{\nu]\alpha}+R\,h_{\mu\nu}\Bigr)+\frac{\Lambda}{3} h_{\mu\nu}\,,
\ee 
obeying 
\be \label{eq: current relation}
\nabla^\mu  J^h_{\mu\nu}=J^h_\nu\,,
\ee 
with the ``{\em secondary current'}
$J_\mu$ given by 
\be \label{Current1form}
J^h_\mu=(G_{\mu\nu}+\Lambda g_{\mu\nu})\,\xi^\nu\,,
\ee 
where $\xi$ is the vector field \eqref{xi} and $G_{\mu\nu}$ is the Einstein tensor. The terms proportional to the cosmological constant $\Lambda$ are necessary for finite results in the case when the spacetime is asymptotically (anti-)de Sitter.

Although $\xi$, in general, need not be a Killing vector, the resultant current $J^h_\mu$ is nevertheless conserved:
\be 
\nabla_\mu J_h^\mu=0\,.
\ee 
This can be used to define {\em gravitational  ``Penrose'' charges}
\begin{equation}\label{eq: Penrose Charge}
    {q_{\mbox{\tiny P}}[h]}=\frac{1}{8\pi}\int_{S^2} {(\star {J}^h)_{\mu\nu}}\,\dd S^{\mu\nu}\,,
\end{equation}
see~\cite{Hull:2025ivk}.
Here $S^2$ is a $2$-dimensional surface which e.g. in the  Boyer--Lindquist-type coordinates is taken to be a sphere of constant radius.
Notice, that from \eqref{eq: current relation} this charge is not necessarily conserved unless $J^{h}_\mu$ vanishes. However, for the spacetimes we consider it turns out that the radial component does vanish, which is enough for the desired conservation. 

In Sec.~\ref{Sec:SpecialCases} we shall give explicit  examples of these charges for the two CKY tensors, namely $h$ and ${k=\star h}$ in the Kerr--BR and  Schwarzschild--BR spacetimes  (in which case indeed $\xi$ and $\zeta$ are not  Killing vectors).

Finally, having a CKY tensor, its square 
\begin{equation}\label{cKT}
K^{\alpha\beta}=h^{\alpha\kappa}{h^{\beta}}_{\kappa}\;,
\end{equation} 
is a {\em conformal Killing tensor}, obeying 
\begin{equation}
\nabla_{\!(\alpha} K_{\beta\gamma)}=g_{(\alpha\beta}\,K_{\gamma)}\;,
\end{equation}
where 
\begin{equation}
K_{\alpha}=\frac{1}{D+2}(2\nabla_{\!\kappa}K^{\kappa}_{\ \alpha}+\nabla_{\!\alpha}K^{\kappa}_{\ \kappa})\,.
\end{equation}
When $h$ in \eqref{cKT} is a KY tensor, then $K_\alpha$ vanishes, and we recover the standard Killing tensor. Such (conformal) Killing tensors give rise to quadratic in momenta constants of motion for (null) geodesics. 
More specifically, for a (null) geodesic with tangent vector $v^\mu=\frac{dx^\mu}{d\lambda}$, where $\lambda$ is an ``affine parameter'' for which $v^\nu\nabla_\nu v^\mu=0$, we obtain the following generalized {\em Carter's constant} of motion:
\be \label{Carter1}
C=K^{\mu\nu} v_\mu v_\nu\,.
\ee 
We refer the interested reader to the review \cite{Frolov:2017kze} and the references therein for more details {on hidden symmetries and their applications in black hole physics.}

\section{Two classes of type D spacetimes with aligned electromagnetic fields}\label{Sec:Carter&PD}

As mentioned in the introduction, there are two related classes of type D spacetimes with aligned electromagnetic fields. 
In this section, we briefly review these classes and discuss their (hidden) symmetries, following largely \cite{Kubiznak:2007kh}.

\subsection{Carter spacetime}
Perhaps best known is the Carter's (often referred to as the Kerr--NUT--AdS) spacetime. 
The metric takes the following form:
\ba
    g^{\mbox{\tiny (C)}}&=&-\dfrac{Q}{\rho^2}(\dd \eta-p^2\dd\sigma)^2+\dfrac{\rho^2}{Q}\dd q^2  \nonumber\\
    &&+\dfrac{P}{\rho^2}(\dd\eta+q^2\dd \sigma)^2 +\dfrac{\rho^2}{P}\dd p^2\,,
    \label{Carter}
\ea
where $\rho^2=q^2+p^2$, and the vector potential reads 
\be\label{ACarter}
A^{\mbox{\tiny (C)}}=-\frac{1}{q^2+p^2}\Bigl[{\rm e}\,q\,(\dd \eta -p^2 \dd \sigma)+{\rm g}\,p\,(\dd\eta+q^2 \dd \sigma)\Bigr]\,.
\ee 
Here  ${\rm e}$ and ${\rm g}$ are the electric, magnetic charges, respectively. 
In order to recover an on-shell solution of the Einstein--Maxwell equations with cosmological constant, the two metric functions
\be\label{off-shell} 
Q=Q(q)\,,\quad P=P(p)\,,
\ee
assume a specific form of quartic polynomials, see Appendix~\ref{App:onshell}. 
However, as is well known, e.g., \cite{Kubiznak:2007kh}, the existence of hidden symmetries of this metric, and the associated separability properties, remain valid {\em off-shell}, for general functions $Q$ and $P$ of the form \eqref{off-shell}---these functions hence need not be concretely specified. 
Actually, this will be the case of all (related) metrics discussed in this paper.

The off-shell Carter's spacetime is known to admit the {\em principal tensor}. 
This is a non-degenerate CCKY 2-form $h^{\mbox{\tiny (C)}}$ obeying 
\be\label{PKY} 
\nabla_\alpha h^{\mbox{\tiny (C)}}_{\beta\gamma}=2 g_{\alpha[\beta}\,\xi_{\gamma]}\,,\quad \xi_\alpha^{\mbox{\tiny (C)}}=\frac{1}{3}\nabla_\gamma \,h^{\mbox{\tiny (C)}}{}^\gamma{}_\alpha\,.
\ee 
Note that the exterior part in this case vanishes and so locally $ h^{\mbox{\tiny (C)}}$ can be expressed as the exterior derivative of a one-form $b$, 
\be 
h^{\mbox{\tiny (C)}}=\dd b\,,\label{cKY_rel}
\ee 
where 
\be 
2b=(p^2-q^2)\dd\eta+p^2q^2\dd\sigma\,,
\ee 
or more explicitly, 
\be \label{PKYexplicit}
h^{\mbox{\tiny (C)}} = 
  -q\,\dd q \wedge (\dd \eta-p^2\dd \sigma)
  +p\,\dd p \wedge (\dd\eta+q^2\dd \sigma)\,.
\ee 
The associated vector field $\xi^{\mbox{\tiny (C)}}$ given by \eqref{PKY} is Killing, and reads 
\be 
\xi^{\mbox{\tiny (C)}}=\partial_\eta\,,
\ee 
while the dual vector obtained form ${k^{\mbox{\tiny (C)}}=\star h^{\mbox{\tiny (C)}}}$
\be\label{zetaC} 
\zeta_{\alpha}^{\mbox{\tiny (C)}}=\frac{1}{3}\nabla_\gamma\, k^{\mbox{\tiny (C)}}{}^\gamma{}_\alpha\,,
\ee 
automatically vanishes, $\zeta^{\mbox{\tiny (C)}}=0$, on account of 
\ba \label{KYCarter}
k^{\mbox{\tiny (C)}}= p\,\dd q \wedge (\dd \eta-p^2\dd \sigma)
   +q\,\dd p\wedge (\dd \eta+q^2 \dd \sigma)\qquad 
\ea 
being a KY tensor.

The two 2-forms then generate the associated (conformal) Killing tensors. 
Namely, $K_{(h)}^{\mbox{\tiny (C)}}{}^{\alpha\beta}=h^{\mbox{\tiny (C)}}{}^{\alpha\gamma}\,h^{\mbox{\tiny (C)}}{}^\beta{}_\gamma$ is a conformal Killing tensor and  $K^{\mbox{\tiny (C)}}_{(k)}{}^{\alpha\beta}=k^{\mbox{\tiny (C)}}{}^{\alpha\gamma}\,k^{\mbox{\tiny (C)}}{}^\beta{}_\gamma$ is a Killing tensor. 
Explicitly, they have the form 
\ba
K_{(h)}^{\mbox{\tiny (C)}}&=&\,\,\frac{q^2Q}{\rho^2}(\dd\eta-p^2 \dd \sigma)^2
   -\frac{q^2 \rho^2}{Q}\dd q^2\nonumber\\ 
&&\!+\frac{p^2P}{\rho^2}(\dd\eta+q^2\dd\sigma)^2
      +\frac{p^2\rho^2}{P}\dd p^2\,,\label{cKT_h}
\ea  
and 
\ba
K_{(k)}^{\mbox{\tiny (C)}}&=&\,\,\frac{p^2Q}{\rho^2}(\dd\eta-p^2 \dd \sigma)^2
   -\frac{p^2 \rho^2}{Q}\dd q^2       \nonumber\\ 
 &&\!+\frac{q^2P}{\rho^2}(\dd\eta+q^2\dd\sigma)^2
   +\frac{q^2\rho^2}{P}\dd p^2\,.\label{cKT_k}
\ea  
Thanks to the Hodge relation between $h$ and $k$, the two are related as follows: 
\ba 
K_{(k)}^{\mbox{\tiny (C)}}&=&K_{(h)}^{\mbox{\tiny (C)}}-\frac{1}{2}h_{\alpha\beta}h^{\alpha\beta} g^{\mbox{\tiny (C)}}\nonumber\\
&=&K_{(h)}^{\mbox{\tiny (C)}}-
(p^2-q^2)g^{\mbox{\tiny (C)}}\,.\label{cKTs_rel}
\ea

Moreover, despite $\zeta$ vanishing, the hidden symmetry also gives rise to the second Killing vector of the spacetime by the following construction:
\be\label{chi} 
\chi^{\mbox{\tiny (C)}}{}^\alpha=K_{(k)}^{\mbox{\tiny (C)}}{}^{\alpha \beta} \xi^{\mbox{\tiny (C)}}_\beta\,.
\ee 
Specifically, one finds
\be \label{chiC}
\chi^{\mbox{\tiny (C)}}=\partial_\sigma\,. 
\ee 
The fact that, apart from the Killing tensor \eqref{cKT_k}, the principal tensor $h^{\mbox{\tiny (C)}}$ also generates the whole ``{\em Killing tower}'' of isometries of $g^{\mbox{\tiny (C)}}$
through \eqref{PKY} and \eqref{chi} is quite remarkable. 
The situation is even more interesting in higher dimensions, where  the above construction of explicit and hidden symmetries from the PKY tensor generalizes to Kerr--NUT--AdS spacetimes in any number of dimensions, see \cite{Krtous:2006qy} for details; the associated separability properties are summarized in \cite{Frolov:2017kze}.

For later purposes, let us finally list two important formulas for the determinant and the Ricci scalar of the (off-shell) Carter spacetime, namely
\ba
\sqrt{-\det g^{\mbox{\tiny (C)}}}&=&\rho^2=q^2+p^2\,,\nonumber\\
R^{\mbox{\tiny (C)}}&=&-\frac{Q''+\ddot P}{\rho^2}\,,
\ea 
where we have denoted ${'= \frac{d}{dq}}$ and ${\,\dot{} =\frac{d}{dp}}$.

\subsection{Pleba\'nski--Demia\'nski spacetime} 

The {\em Pleba\'nski--Demia\'nski} (PD) solution \cite{Plebanski:1976gy} is the most general type D solution of Einstein--Maxwell equations with cosmological constant and \emph{aligned}%
\footnote{\label{FN1}Here aligned can also be understood as the fact that the eigenvectors of the field tensor are the same as for the principal tensor.}
electromagnetic field. 
The solution takes the following form:
\be\label{PD} 
g^{\mbox{\tiny (PD)}}=\frac{1}{\Omega_{\mbox{\tiny (PD)}}^2}\,g^{\mbox{\tiny (C)}}\,,\quad 
A^{\mbox{\tiny (PD)}}=A^{\mbox{\tiny (C)}}\,.
\ee
Here $g^{\mbox{\tiny (C)}}$ is the Carter's metric \eqref{Carter}, the conformal factor $\Omega_{\mbox{\tiny (PD)}}$ takes the following particular  form:
\be\label{OmegaPD} 
\Omega_{\mbox{\tiny (PD)}}=1-pq\,,
\ee 
and the vector potential $A^{\mbox{\tiny (PD)}}$ remains formally identical to $A^{\mbox{\tiny (C)}}$, given by  \eqref{ACarter}.
In order this to be indeed a solution of vacuum Einstein--Maxwell equations with cosmological constant, the metric functions $P$ and $Q$ must take a specific form, described in Appendix~\ref{App:onshell}.
However, as mentioned above, the existence of hidden symmetries and their implications remain valid for the off-shell PD spacetime \eqref{PD}, with arbitrary functions $Q$ and $P$ \eqref{off-shell}.

Following \cite{Kubiznak:2007kh}, let us review hidden symmetries of the PD spacetime. 
The spacetime no longer possesses the PKY tensor---only a weaker structure of a CKY 2-form is present---essentially following from the conformal property \eqref{Weyl} and \eqref{conformalProp}. 
Namely, we have the following CKY 2-form:
\be \label{hPD}
h^{\mbox{\tiny (PD)}}=\frac{1}{\Omega_{\mbox{\tiny (PD)}}^3}\, h^{\mbox{\tiny (C)}}\,,
\ee 
obeying \eqref{CKY4D} with $\xi^{\mbox{\tiny (PD)}}$ given by \eqref{xi}. The latter turns out to be a  Killing vector and reads 
\be 
\xi^{\mbox{\tiny (PD)}}=\partial_\eta\,.
\ee 
Moreover, the Hodge dual of $h$ (with respect to the metric $g^{\mbox{\tiny (PD)}}$) is a dual CKY 2-form and reads 
\be 
k^{\mbox{\tiny (PD)}}=\star h^{\mbox{\tiny (PD)}}=\frac{1}{\Omega^3_{\mbox{\tiny (PD)}}} \,k^{\mbox{\tiny (C)}}\,.
\ee 
It gives rise to the dual Killing vector $\zeta$ via \eqref{zeta}, which is now non-trivial and reads 
\be\label{zetaPD} 
\zeta^{\mbox{\tiny (PD)}}=\partial_\sigma\,.
\ee 

The above CKY 2-forms ``square to'' the corresponding conformal Killing tensors via \eqref{cKT}. 
Namely, we find the following for covariant indices:
\be 
K^{\mbox{\tiny (PD)}}_{(h)\alpha\beta}=\frac{1}{\Omega_{\mbox{\tiny (PD)}}^4}\,K^{\mbox{\tiny (C)}}_{(h)\alpha\beta}\,,
\quad 
K^{\mbox{\tiny (PD)}}_{(k)\alpha\beta}=\frac{1}{\Omega_{\mbox{\tiny (PD)}}^4}\,K^{\mbox{\tiny (C)}}_{(k)\alpha\beta}\,,
\ee 
in terms of the (conformal) Killing tensors of the Carter spacetime \eqref{cKT_h} and \eqref{cKT_k}. 
Moreover, they are all the same when written as contravariant objects, that is,  
\begin{align}\label{eq: CKT h explicit}
    K^{\mbox{\tiny (PD)}}_{(h)} &=
     K^{\mbox{\tiny (C)}}_{(h)}
    \nonumber
    \\
    &= \ q^2\Big[\,\frac{1}{\rho^2Q}(q^2\partial_\eta-\partial_\sigma)^2-\frac{Q}{\rho^2}\partial_q^2\Big]
    \nonumber\\
    &\ \,+p^2\Big[\frac{1}{\rho^2P}(p^2\partial_\eta+\partial_\sigma)^2+\frac{P}{\rho^2}\partial_p^2\Big],
\end{align}
and
\begin{align}\label{eq: CKT k explicit}
    K^{\mbox{\tiny (PD)}}_{(k)}&=K^{\mbox{\tiny (C)}}_{(k)}
   \nonumber\\
   &=\  p^2\Big[\,\frac{1}{\rho^2Q}(q^2\partial_\eta-\partial_\sigma)^2-\frac{Q}{\rho^2}\partial_q^2\Big]
    \nonumber\\
    &\ \,+q^2\Big[\frac{1}{\rho^2P}(p^2\partial_\eta+\partial_\sigma)^2+\frac{P}{\rho^2}\partial_p^2\Big].
\end{align}
Note that there is no contradiction with the previous expressions because the indices are raised and lowered with the different metrics $g^{\PD}$ and $g^{\C}$.

The two obey 
\be \label{Qk-PD}
K^{\mbox{\tiny (PD)}}_{(k)}=K^{\mbox{\tiny (PD)}}_{(h)}-\frac{p^2-q^2}{\Omega_{\mbox{\tiny (PD)}}^2}\,g^{\mbox{\tiny (PD)}}\,.
\ee 
Finally, the construction {\`a} la  \eqref{chi} now yields 
\be \label{chiPD}
\chi^{\mbox{\tiny ({PD})}}{}^\alpha=K_{(k)}^{\mbox{\tiny (PD)}}{}^{\alpha \beta} \xi^{\mbox{\tiny (PD)}}_\beta=\frac{1}{\Omega^2_{\mbox{\tiny (PD)}}}(\partial_\sigma)^\alpha\,, 
\ee 
which is no longer a (conformal) Killing vector.
Thus, we conclude that also in the PD case both isometries of the spacetime follow from the existence of the hidden symmetry, although the construction is different from that in the Carter spacetime. Namely, the ``Killing tensor construction'' \eqref{chi} is replaced by the ``dual vector construction'' \eqref{zetaPD}.

\section{Novel family of black hole spacetimes and its hidden symetries}
\label{Sec:OP}

Having reviewed the two well-known cases, let us now turn to the novel family of type D metrics with non-aligned electromagnetic fields, namely the {\em Ovcharenko and Podolsk\'y (OP)} general family of spacetimes \cite{Podolsky:2025tle, Ovcharenko:2025cpm} and discuss their (hidden) symmetries.

\subsection{Ovcharenko--Podolsk\'y spacetime}

The new family of spacetimes can be written in the following general form \cite{Podolsky:2025tle, Ovcharenko:2025cpm}:
\be\label{OP} 
g^{\mbox{\tiny (OP)}}=\frac{1}{\Omega_{\mbox{\tiny (OP)}}^2}\,g^{\mbox{\tiny (C)}}\,,
\ee
where, again,  $g^{\mbox{\tiny (C)}}$ is the Carter metric \eqref{Carter}. 
The conformal factor $\Omega_{\mbox{\tiny (OP)}}$ takes the form
\begin{align}\label{OmegaOP}
\Omega^2_{\mbox{\tiny (OP)}}=&\ 1 - 2\,p\,q + c_{10}\,p +c_{01}\,q +c_{02}\,(q^2-p^2) \nonumber\\
    &+c_{12}\,p\,q^2 +c_{21}\,p^2q+c_{22}\,p^2q^2\,,
\end{align}
where $\{ c_{10}, c_{01},c_{02},c_{12}, c_{21}, c_{22}\}$ are 6 real parameters. 
The geometry is supported by the vector potential
\ba\label{A_OP}
A^{\mbox{\tiny (OP)}}&=&\frac{1}{q^2+p^2}\Bigl[
    ({\rm g}\,q-{\rm e}\,p)\,\Omega_{,p}^{\mbox{\tiny (OP)}}(\dd\eta-p^2 \dd \sigma) \nonumber\\
&& \qquad\quad + ({\rm e}\,q+{\rm g}\, p)\,\Omega^{\mbox{\tiny (OP)}}_{,q}(\dd \eta +q^2 \dd \sigma)
\Bigr]\nonumber\\
&& \qquad\quad -{\rm e}\,\Omega^{\mbox{\tiny (OP)}} \dd \sigma\,,
\ea 
describing the electromagnetic field with non-aligned components. This field is  characterized by two charge parameters  ${\rm e}$ and~${\rm g}$, related to the complex parameter $\bar c'$ presented in \cite{Ovcharenko:2025cpm} by $1/{\bar c'}=2({\rm e}+\im\, {\rm g})$. In order to satisfy the system of Einstein--Maxwell field equations, the metric functions $Q$ and $P$ take the form of quartic polynomials (see Appendix~\ref{App:onshell}). 
However, as in the aligned electromagnetic case, the existence of hidden symmetries is {\em independent of the explicit form} of $Q$ and $P$, and we can thus consider the general off-shell case \eqref{off-shell}.

\subsection{Hidden symmetries}

Similar to the PD spacetime, the OP spacetime does not admit the full PKY tensor, and only a weaker structure of the CKY 2-form remains present. 
Namely, it takes the following explicit form:
\be \label{CKY-OP-h}
h^{\mbox{\tiny (OP)}}=\frac{1}{\Omega_{\mbox{\tiny (OP)}}^3} \,h^{\mbox{\tiny (C)}}\,,
\ee 
with $h^{\mbox{\tiny (C)}}$ given in \eqref{PKYexplicit}. Its dual is also a CKY 2-form and reads 
\be \label{CKY-OP-k}
k^{\mbox{\tiny (OP)}}=\star h^{\mbox{\tiny (OP)}}=\frac{1}{\Omega^3_{\mbox{\tiny (OP)}}} \,k^{\mbox{\tiny (C)}}\,.
\ee 
Together they give rise to the associated covariant conformal Killing tensors 
\be 
K^{\mbox{\tiny (OP)}}_{(h)}=\frac{1}{\Omega_{\mbox{\tiny (OP)}}^4}\,K^{\mbox{\tiny (C)}}_{(h)}\,,
\quad 
K^{\mbox{\tiny (OP)}}_{(k)}=\frac{1}{\Omega_{\mbox{\tiny (OP)}}^4}\,K^{\mbox{\tiny (C)}}_{(k)}\,,
\ee 
related by 
\be 
K^{\mbox{\tiny (OP)}}_{(k)}=K^{\mbox{\tiny (OP)}}_{(h)}-\frac{p^2-q^2}{\Omega_{\mbox{\tiny (OP)}}^2}\,g^{\mbox{\tiny (OP)}}\,.
\ee 
These expressions are fully analogous to \eqref{hPD} and\eqref{Qk-PD} for the PD solution, and their contravariant expressions are identical to \eqref{eq: CKT h explicit} and \eqref{eq: CKT k explicit} respectively.

Interestingly, and contrary to the Carter and PD spacetimes, the above hidden symmetry \emph{no longer} seems to imply the existence of a {\em Killing tower of symmetries}. 
More specifically,  although the two isometries, $\partial_\eta$ and $\partial_\sigma$, still exist in the spacetime $g^{\mbox{\tiny (OP)}}$, we currently do not know how to ``generate'' them from the existence of the above two CKY 2-forms. 
This is closely related to the particular form of the conformal factor $\Omega_{\mbox{\tiny (OP)}}$ given by \eqref{OmegaOP} which, in its turn, corresponds to the ``non-alignment'' of the electromagnetic field. 
To see this in more detail, let us now consider even more general ``conformal-to-Carter spacetimes'' and study the corresponding candidates on Killing vectors.

\renewcommand{\arraystretch}{1.5}

\begin{center}
\setlength{\tabcolsep}{6pt}
\begin{table*}[t!]

\begin{tabular}{ |c||c|c|c|p{2.5cm}|  }
 \hline
 \multicolumn{5}{|c|}{$\text{Killing vector of the form}\quad a\,\partial_\eta+b\,\partial_\sigma$ } \\
 \hline
 Vector & ${\,a\neq0\neq b\,}$  &$a=0\,, b\neq 0$& $a\neq0\,, b=0$ & ${\quad a=0=b}$\\
 \hline
 \hline
 $\xi=\frac{1}{3}\nabla\cdot h$  & \xmark    &\xmark&   $\Omega=a+c\,pq$ \ (PD or Carter)& $\Omega=pq$ \   (Carter) \\
 $\zeta=\frac{1}{3}\nabla\cdot k$  & \xmark    & $\Omega=c+b\,pq$ \ (PD or Carter)&  \xmark  & $\Omega=c$  \quad (Carter) \\
 $\eta=K_k\cdot \xi$  & \xmark    & $\Omega=\dfrac{pq}{c+b\, pq }$\ (N or Carter)&   \xmark & $\Omega=pq$ \ (Carter) \\
  $\hat{\eta}=K_h\cdot \zeta$  & \xmark    &\xmark&   $\Omega=\dfrac{pq}{a+c\,pq} $ \ (N or Carter)& $\Omega=c$ \quad (Carter) \\
$\iota=K_k\cdot \zeta$  & \xmark    &\xmark&   \xmark & $\Omega=c$ \quad (Carter)\\
$\hat{\iota}=K_h\cdot \xi$  & \xmark    &\xmark&  \xmark & $\Omega=c$ \quad (Carter)\\
 \hline
\end{tabular}
\caption{
{\bf Construction of the Killing vectors for the conformal-to-Carter spacetimes,} $g=\Omega^{-2}g^\C$. 
This table summarizes various conformal factors $\Omega$ (and the associated spacetimes) that give rise to Killing vector fields from the existence of a hidden symmetry of the CKY 2-form ${h=\Omega^{-3}h^{\C}}$ and its dual ${k=\star h=\Omega^{-3}k^{\C}}$.
Here, $a$ and $b$ are the constants determining the specific Killing vector of the form ${\,a\,\partial_\eta+b\,\partial_\sigma\,}$, while $c$ is an arbitrary constant of integration.
The symbol \xmark\  indicates that the solution does not exist.
Since the OP spacetimes  (including the Kerr--BR and Schwarzschild--BR black holes) are not in this table, this proves the claim in Sec.~\ref{Sec:OP} that although they admit two Killing vectors ${\partial_\eta}$ and ${\partial_\sigma}$, these cannot be generated (in a straightforward manner, at least) from their two CKY 2-forms \eqref{CKY-OP-h} and \eqref{CKY-OP-k}. 
Especially interesting is the novel (N) conformal factor \eqref{novel}, where the full set of isometries can be generated, yet in a completely different fashion to what happens in the Carter and Pleba\'nski--Demia\'nski (PD) spacetimes.
}
\label{Tab: conform factors}
\end{table*}
\end{center}

\section{Conformal-to-Carter spacetimes and the Killing tower}\label{Sec:Tower}

Let us consider the following {\em more genaral off-shell spacetime}
\be
g=\frac{1}{\Omega^2}\,g^{\mbox{\tiny (C)}}\,,\label{off-shell-PD}
\ee 
where $g^{\mbox{\tiny (C)}}$ is the Carter off-shell spacetime \eqref{Carter} and $\Omega$ is an ``arbitrary'' conformal factor which preserves the two isometries of Carter's spacetime, $\partial_\eta$ and $\partial_\sigma$. 
That is, we have   
\be \label{ctCoffshell}
\Omega=\Omega(p,q)
\,,\quad P=P(p)\,,\quad Q=Q(q)\,.
\ee 
Obviously, we recover the Carter, PD, and OP spacetimes by choosing $\Omega=1, \Omega_{\mbox{\tiny (PD)}}$, and $\Omega_{\mbox{\tiny (OP)}}$ given by (\ref{OmegaPD}) and (\ref{OmegaOP}), respectively. 
As we shall see, the separability properties discussed in the subsequent section remain valid {\em irrespective of the specific choice of\,} $\Omega$, and derive from the existence of hidden symmetries inherited from the (conformally related) Carter ``seed'' spacetime (\ref{Carter}), {as well as the fact that the metric determinant and the Ricci scalar are
\be 
\sqrt{-\det g}=\frac{1}{\Omega^4} \sqrt{-\det g^{\mbox{\tiny (C)}}}=
\frac{\rho^2}{\Omega^4}\,,
\ee
and 
\be 
R=\Omega^{2}\bigl(R^{\mbox{\tiny (C)}}-6\,\Omega\,\Box^{\mbox{\tiny (C)}} \Omega^{-1}\bigr)\,.
\ee 
Note that these are slightly different expressions from conventional expressions (e.g. ~\cite[Appendix D]{Wald1984general}) because of the inverse scaling in $g$.
}

\subsection{Hidden symmetries of the conformal spacetime}

The general conformal metric (\ref{off-shell-PD}) admits the following CKY 2-forms: 
\be 
h=\frac{1}{\Omega^3}\, h^{\mbox{\tiny (C)}}\,,\quad 
k=\star h=\frac{1}{\Omega^3} \,k^{\mbox{\tiny (C)}}\,,
\ee 
where $h^{\mbox{\tiny (C)}}$ is the principal tensor of Carter's spacetime \eqref{PKYexplicit} and $k^{\mbox{\tiny (C)}}$ is the dual KY tensor \eqref{KYCarter}. 
The two give rise to the associated covariant conformal Killing tensors
\be \label{eq: CKTs off shell}
K_{(h)}=\frac{1}{\Omega^4}\,K^{\mbox{\tiny (C)}}_{(h)}\,,
\quad 
K_{(k)}=\frac{1}{\Omega^4}\,K^{\mbox{\tiny (C)}}_{(k)}\,,
\ee 
related by 
\be \label{eq: CKT rel}
K_{(k)}=K_{(h)}-\frac{p^2-q^2}{\Omega^2}\,g\,.
\ee 
The contravariant expressions for $K_{(h)}$ and $K_{(k)}$ are also given by the right hand sides of \eqref{eq: CKT h explicit} and \eqref{eq: CKT k explicit} respectively.

The above spacetime can be endowed with a {\em test aligned} electromagnetic field, obeying 
\be 
\nabla_\mu F^{\mu\nu}=0\,,\quad F=\dd A\,,
\ee 
provided we choose 
\be 
A=A^{\mbox{\tiny (C)}}\,.
\ee 
Note, however, that it may not always be possible to backreact this field on the geometry, by choosing the explicit expressions for the functions $Q$ and $P$.

\subsection{Candidates on Killing vectors}

Let us now try to generate the {\em Killing tower} of isometries from the above hidden symmetry. 
To start with, let us consider the associated with $h$ vector $\xi$, \eqref{xi}, namely 
\be 
\xi^\alpha=\frac{1}{3}\nabla_\beta h^{\beta\alpha}\,.
\ee 
As noted in Sec.~\ref{Sec:HS},  such a vector is guaranteed to satisfy the following equation:  
\be 
\nabla_{(\alpha}\xi_{\beta)}=\frac{1}{2}h_{(\alpha|\gamma|}R^\gamma{}_{\beta)}\,.
\ee 
Obviously, in an Einstein space the R.H.S. vanishes, and $\xi$ is a Killing vector. 
More interestingly, the R.H.S. vanishes in the PD spacetime, where the Ricci tensor is determined from the aligned electromagnetic field and the corresponding conformal factor takes the specific form $\Omega^{\mbox{\tiny (PD)}}$, given by \eqref{OmegaPD}. 
As mentioned in Footnote \ref{FN1}, this is because the principal tensor is also aligned with these null directions, so for such spacetimes these will anti-commute as matrices, see e.g. the discussion in \cite{Frolov:2017bdq}.

In the off-shell conformal spacetime, however, this is not the case, and we find
\be 
\xi=\Bigl(\Omega-\frac{p^3\, \Omega_{,p}+q^3\,\Omega_{,q}}{\rho^2}\Bigr)\partial_\eta
+\frac{q\,\Omega_{,q}-p\,\Omega_{,p}}{\rho^2}\,\partial_\sigma\,.
\ee 
Demanding that this is a Killing vector, i.e., of the form 
\be\label{KVwant} 
\mbox{KV}=a\,\partial_\eta+b\,\partial_\sigma\,,
\ee 
for some constants $a$ and $b$, yields in general an inconsistent system of two PDEs for the conformal factor $\Omega$. 
Namely, solving the two equations separately yields
\be 
\Omega=a+pq\, F\Bigl(\frac{p^2-q^2}{p^2q^2}\Bigr)\,,\quad 
\Omega=\frac{b}{2}(q^2-p^2)+\tilde F(pq)\,,\label{Om_integr}
\ee
with $F$ and $\tilde F$ arbitrary functions. 
Obviously, for both non-trivial $a$ and $b$ there is no solution. 

When ${b=0}$, the solution reads
\be
{\Omega=a+c\,pq}\,,
\ee 
where ${c=\mbox{const.}}$ is an integration constant, seemingly generalizing the PD case. 
However, it can be shown, see {\bf Lemma 1} in Appendix~\ref{App:B},  that when both $a$ and $c$ are non-trivial, one can rescale the coordinates so that such $\Omega$ can be brought into the PD form, \eqref{OmegaPD}.
When $a=0$ and $b\neq 0$,  we find that the system is inconsistent. 
Finally, when $a=0=b$, we find (up to a constant scale that can be dropped)
\be 
\Omega=pq\,.
\ee 
Although $\xi$ is not a Killing vector as it vanishes in this case, it means that with this choice of conformal factor, the above $h$ is actually a KY tensor (obtained from the PKY tensor $h^{\mbox{\tiny (C)}}$ by the appropriate conformal transformation) and its dual is the PKY 2-form. 
Thus, by the uniqueness result \cite{Krtous:2008tb}, the spacetime has to be again Carter's spacetime, just written in some different coordinate system. 
That this is indeed the case is shown by {\bf Lemma 2} in  Appendix~\ref{App:B}, see also \cite{Kubiznak:2007kh}.

Similarly, we may ask when the dual vector $\zeta$, \eqref{zeta},
\be 
\zeta^{\alpha}=\frac{1}{3}\nabla_{\!\beta}k^{\beta\alpha}\;,
\ee
is a Killing vector. 
We find
\be 
\zeta=\frac{pq\,(p\,\Omega_{,q}-q\,\Omega_{,p})}{\rho^2}\partial_\eta+\frac{q\,\Omega_{,p}+p\,\Omega_{,q}}{\rho^2}\partial_\sigma\,.
\ee 
Equating this to \eqref{KVwant} separately yields
\be 
\Omega=\frac{a}{2}\frac{p^2-q^2}{pq}+F(pq)\,,\quad \Omega=b\, pq+\tilde F(pq)\,.
\ee 
Obviously, for $a$ and $b$ non-trivial, the system is again inconsistent. 
When $b=0$ and $a\neq 0$ it remains inconsistent. 
When $a=0$ and $b$ is non-trivial, we find ${\Omega=c+b\,pq}$, where ${c=\hbox{const.}}$ is an integration constant. 
So, upon the due coordinate transformation, we are back to Carter's spacetime (for ${c=0}$) or the Pleba\'nski--Demia\'nski case ${(c\neq 0)}$.
Finally, when both ${a=0=b}$, we recover a constant conformal factor, returning back to the Carter class of spacetimes.%
\footnote{Note that the above results show that it is impossible to generate a new PKY tensor from the original Carter's PKY tensor  by conformal transformation---a fact related to the uniqueness theorem proved in \cite{Krtous:2008tb}.}

Let us next consider a vector field 
\be 
\eta^\alpha=K^{\alpha \beta}_{(k)}\,\xi_\beta\,,
\ee 
which would yield the second Killing vector, $\partial_\sigma$, in the Carter case. 
In our case, however, we find 
\be 
\eta=-\frac{p^2q^2(p\,\Omega_{,p}-q\,\Omega_{,q})}{\Omega^2\rho^2}\,\partial_\eta+\Bigl(\frac{1}{\Omega}-\frac{pq^2\Omega_{,p}+qp^2\Omega_{,q}}{\Omega^2\rho^2}\Bigr)\partial_\sigma\,.
\ee
Equating this with a Killing vector of the form \eqref{KVwant}, the system of PDEs is only consistent when $a=0$ and $b\neq 0$, in which case we have 
\be \label{Omega-eta}
\Omega=\frac{pq}{c+b\, pq}\,,
\ee 
where $c$ is an integration constant. 
This is a rather interesting conformal factor---we shall briefly return to the corresponding spacetime at the end of this section. 
Finally, when ${a=0=b}$, in which case no Killing vector is generated, we have  $\Omega=pq$ and return to Carter's case.

Similarly, we may consider the following vector:
\be 
\hat \eta^\alpha=K^{\alpha\beta}_{(h)}\,
\zeta_\beta\,.
\ee 
Explicitly, we find: 
\be 
\hat \eta=\frac{pq\,(p^3\Omega_{,p}+q^3\Omega_{,q})}{\rho^2\Omega^2}\,\partial_\eta+
\frac{pq\,(p\,\Omega_{,p}-q\,\Omega_{,q})}{\rho^2\Omega^2}\,\partial_\sigma\,.
\ee 
Compared to the desired form of the Killing vector, \eqref{KVwant}, we find that the solution is only consistent when $b=0$. When $a\neq 0$, we have 
\be \label{Omega-hateta}
\Omega=\frac{pq}{a+c\,pq}\,,
\ee 
with $c$ an integration constant. We also find ${\Omega=\hbox{const.}}$ when both ${a=0=b}$.

Turning next to
\be 
\iota^\alpha=K^{\alpha\beta}_{(k)}\,\zeta_\beta\,,
\ee 
we find: 
\be 
\iota=\frac{p^2q^2(q\,\Omega_{,p}+p\,\Omega_{,q})}{\rho^2\Omega^2}\,\partial_\eta+
\frac{q^3\Omega_{,p}-p^3\Omega_{,q}}{\rho^2\Omega^2}\,\partial_\sigma\,.
\ee 
It only yields the Killing vector of the form \eqref{KVwant} when both ${a=0=b}$, in which case ${\Omega=\hbox{const.}}$ and the vector vanishes.

Let us finally consider
\be 
\hat \iota^\alpha=K^{\alpha\beta}_{(h)}\,\xi_\beta\,.
\ee 
Then we find 
\ba 
\hat \iota&=&\frac{\Omega(p^4-q^4)+q^5\Omega_{,q}-p^5\Omega_{,p}}{\rho^2\Omega^2}\,\partial_\eta\nonumber\\
&&+\frac{\rho^2\Omega-q^3\Omega_{,q}-p^3\Omega_{,p}}{\rho^2\Omega^2}\,\partial_\sigma\,,
\ea 
which is of the form \eqref{KVwant} only when ${a=0=b}$, in which case ${\Omega=pq}$ (Carter's spacetime) and the vector vanishes. 

The above calculations are summarized in Table~\ref{Tab: conform factors}. 
From here it is obvious  that for $\Omega$ quadratic in $p$ and $q$, as in \eqref{OmegaOP},  
none of the above vector fields are Killing vectors. 
Thus, the existence of a non-aligned electromagnetic field in the OP spacetimes, encoded in the complicated form of the conformal factor $\Omega^{\mbox{\tiny (OP)}}$, given by \eqref{OmegaOP}, seems to spoil the Killing tower property which exists for both the Carter and PD spacetimes. 
Whether an appropriate generalization of such a construction can be found for the spacetimes with non-aligned electromagnetic field remains to be seen in the future.

\subsection{Novel spacetime}

We have discovered that for the conformal factor 
\be\label{novel} 
\Omega=\frac{pq}{B+A\,pq}\,,
\ee 
where $A$ and $B$ are constants, we {\em have the full Killing tower}. 
Specifically, we find the Killing vectors
\be 
\hat \eta=B\,\partial_\eta\,,\quad \eta=A\,\partial_\sigma\,,
\ee 
while the other vector fields are (for ${A\neq0\neq B}$) not Killing. 
This is thus a very different construction from both the Carter one and the PD one.
It can be shown, see {\bf Lemma 3} in Appendix~\ref{App:B}, that when both $A$ and $B$ are non-trivial, the above conformal factor can be brought to the following canonical form:   
\be 
\Omega=\frac{pq}{1-pq}\,.
\ee 
While we have checked that by choosing $Q=Q(q)$ and $P=P(p)$ such a spacetime cannot be a vacuum solution of Einstein equations, it may (perhaps) describe some interesting solution with matter that admits a hidden symmetry of the CKY 2-form together with the full tower of Killing vector fields derived from it. 
Such a {\em novel spacetime} (N) is worth investigating in the future; for now, we turn to the separability of test particles and fields in conformal-to-Carter spacetimes.

\section{Separability structures}\label{Sec:Separability}

{In this section we shall demonstrate separability of various test field or particle equations in the general off-shell conformal-to-Carter spacetimes \eqref{off-shell-PD}, \eqref{ctCoffshell}. 
We note that similar observations, though with different motivations and a different form of conformal factor, were  made for example in \cite{Hamilton:2022sgf}. 
See also \cite{Acevedo:2025zrq} for a generalization beyond the Carter spacetime. 
}

\subsection{{Massless} Hamilton--Jacobi equation}

The above {hidden symmetries} can be used for the separation of variables in the massless Hamilton--Jacobi equation
\begin{align}
    g^{\mu\nu}\dfrac{\partial S}{\partial x^{\mu}}\dfrac{\partial S}{\partial x^{\nu}}=0\,,\label{HJ_eq_1}
\end{align}
where $S$ is Hamilton's principal function.%
\footnote{In principle, one might also consider a modified ``conformal Hamilton--Jacobi equation''
\be 
 g^{\mu\nu}\dfrac{\partial S}{\partial x^{\mu}}\dfrac{\partial S}{\partial x^{\nu}}+\frac{1}{6} R=0\,,
\ee 
considered in \cite{Gray:2021wzf}.
However, such an equation neither is truly conformally invariant, nor does it separate in our spacetime, that is why we do not consider it in our paper. 
The reason as to why such an equation does not separate can be traced to the presence of an {\em ${\cal R}$-factor}  for the conformal wave equation, see \eqref{eq: R sep Ansatz} below. 
}

First, the existence of the 2 Killing vectors $\partial_{\eta}$ and $\partial_{\sigma}$ ensures that the following quantities: 
\begin{align}\label{eq: Killing consts}
    (\partial_\eta)^{\mu}\dfrac{\partial S}{\partial x^{\mu}}=-E\,,~~~(\partial_\sigma)^{\mu}\dfrac{\partial S}{\partial x^{\mu}}=L\,,
\end{align}
are conserved. 
Moreover, the conformal Killing tensors \eqref{eq: CKTs off shell} of the off shell spacetime \eqref{off-shell-PD} provide two additional constants of motion.
However, due to the relation \eqref{eq: CKT rel} and the massless Hamilton--Jacobi equation \eqref{HJ_eq_1}, they are not independent. 
Therefore, we consider only {\em Carter's constant} \eqref{Carter1} associated with $h$, i.e.
\begin{equation}\label{Cexplicit}
 C=K_{(h)}^{\mu\nu}\,\dfrac{\partial S}{\partial x^{\mu}}\dfrac{\partial S}{\partial x^{\nu}}\,.
\end{equation}

The fact that these constants of motion are at most quadratic in momentum, and in involution with respect to the Poisson bracket (see \cite{Frolov:2017kze} and \cite{Gray2025} for details on the integrability), allow for an additive separation of variables for Hamilton's principal function, namely
\begin{align}\label{separation-ansatz}
S= W_\eta(\eta)+W_\sigma(\sigma)+W_q(q)+W_p(p)\,.
\end{align}
These functions  $W_\mu(x^\mu)$  depend on the one particular coordinate $x^\mu$ ($\mu$ here is just a label not an (abstract) index), as well as, (possibly all of) the constants of motion $\{E,L,C\}$.
Plugging this into \eqref{eq: Killing consts} one immediately obtains
\begin{equation}
    W_\eta=-E \,\eta\,,\qquad W_\sigma=L \,\sigma\,.
\end{equation}

On the other hand, for $W_q$ and $W_p$, making use of {\eqref{Cexplicit} and the explicit expression} \eqref{eq: CKT h explicit}, we have
\begin{align}
    \dfrac{1}{\rho^2}\Big[&\ q^2\Big(\dfrac{(Eq^2+L)^2}{Q}-Q(W_q')^2\Big)\nonumber\\
    +&\ p^2\Big(\dfrac{(Ep^2-L)^2}{P}+P(\dot{W}_p)^2\Big)\Big]=C\,,\label{Cart_const}
\end{align}
which can be rearranged to give
\begin{align}
  P\,(\dot{W}_p)^2&=C-\dfrac{(Ep^2-L)^2}{P}\,,\label{eq: W_p}
  \\
  Q\,(W'_q)^2&=\dfrac{(Eq^2+L)^2}{Q}-C\,.\label{eq: W_q}
\end{align}
These are two first-order decoupled ODEs for $\dot{W}_p$ and $W'_q$.
One can check that \eqref{eq: W_p} and \eqref{eq: W_q} imply the Hamilton--Jacobi equation is satisfied. 
That is, using the form of \eqref{off-shell-PD} we obtain
\begin{align}
 &g^{\mu\nu}\dfrac{\partial S}{\partial x^{\mu}}\dfrac{\partial S}{\partial x^{\nu}}
 &\nonumber\\
  &=\dfrac{(Ep^2-L)^2}{P}-\dfrac{(Eq^2+L)^2}{Q}+Q\,(W'_q)^2+P\,(\dot{W}_p)^2
    \nonumber\\
    &=C-C=0\,.\label{HJ_eq_2}
\end{align}

{The obtained expressions} can be used to find the tangent vector to null geodesics (momentum), upon using  
\be \label{eq: velocity momentum HJ relation}
\dfrac{\partial S}{\partial x^{\mu}}=p_{\mu}=g_{\mu\nu}\frac{dx^\mu}{d\lambda}\,,
\ee
where $\lambda$ is an affine parameter. 
In particular, this yields ${p_{\eta}=-E\,,\,p_{\sigma}=L}$, and one obtains the following  expressions for the geodesic 4-velocity:
\begin{align}
    \frac{\dd\eta}{\dd\lambda}&=\dfrac{\Omega^2}{\rho^2}\Big[\dfrac{q^2(Eq^2+L)}{Q}-\dfrac{p^2(Ep^2-L)}{P}\Big]\,,\nonumber\\
    \frac{\dd\sigma}{\dd\lambda}&=\dfrac{\Omega^2}{\rho^2}\Big[-\dfrac{Eq^2+L}{Q}-\dfrac{Ep^2-L}{P}\Big]\,,\nonumber\\
    \frac{\dd q}{\dd\lambda}&={\pm}\dfrac{\Omega^2}{\rho^2}\sqrt{(Eq^2+L)^2-C Q}\,,\\
    \frac{\dd p}{\dd\lambda}&={\pm}\dfrac{\Omega^2}{\rho^2}\sqrt{CP-(Ep^2-L)^2}\,,\nonumber
\end{align}
where the signs in front of the square roots are independent. 
These expressions can be further decoupled by introducing generalized `Mino time'~\cite{Mino:2003yg}. 
The above expressions can also be used to calculate the black hole shadows~\cite{Perlick:2021aok} in the appropriate on shell spacetimes, e.g. \cite{Wang:2025vsx, Zeng:2025tji, Ali:2025beh, Vachher:2025jsq}.

\subsection{Conformal scalar field equation}

{Let us next turn to the scalar field equation.}
For the metric \eqref{off-shell-PD} the usual massless Klein--Gordon equation $g^{\mu\nu}\nabla_{\mu}\nabla_{\nu}\psi=0$ {does not separate off-shell, when} the Ricci scalar $R\neq 0$.
However, the conformal Klein--Gordon equation
\begin{align}
    \Box_R\,\psi=\Big(g^{\mu\nu}\nabla_{\mu}\nabla_{\nu}-\dfrac{1}{6}R\Big)\psi=0\,,\label{cKG_eq}
\end{align}
is ${\cal R}$-separable --- the scalar field admits a multiplicative separation of variables up to some arbitrary function of the coordinates ${\cal R}$ --- i.e. 
\be \label{eq: R sep Ansatz}
\psi={\cal R}(x)X_\eta(\eta)X_{q}(q)X_{p}(p)X_{\sigma}(\sigma)
\,,
\ee 
where,  if the function $\cal R$ is chosen correctly, each of the functions $X_i$ satisfy a decoupled ODE.
The ${\cal R}$-separability follows from the underlying symmetry operators presented in \cite{Gray:2020rtr,Gray:2021wzf} for the conformal-to-Carter spacetime \eqref{off-shell-PD}.

Generically, a conformal symmetry of operator of $\Box_R$ is an operator $\cal O$, which satisfies 
\be
\Box_R\big( {\cal O}(\psi)\big) =\widetilde{\cal O}\big( \Box_R(\psi)\big) \,,
\ee
for all members of the class of conformally related metrics $[g]=\{\omega^{2} g\,|\,\omega\in C^\infty(M)\}$ and fields $[\psi]=\{\omega^{-1}\psi\,|\, \omega \in C^\infty(M)\}$, and for some related operator $\widetilde{\cal O}$. 
These will preserve the kernel of the conformal wave equation, hence the solution space.

The classification of such operators was carried out in full generality \cite{Michel_2014}, and in particular, for the off-shell Carter metric these operators can be written in terms of a Killing tensor $K$ and Killing vectors $l$~\cite{Gray:2021wzf,Michel_2014}:
\begin{align}\label{eq: L op}
    {\cal L}(\psi)&=l^\alpha\nabla_\alpha\psi+\frac{1}{4}\psi\,\nabla_{\alpha}l^\alpha\,,\\
    \label{eq: K op}
    {\cal K}(\psi)&=\nabla_{\alpha}( K^{\alpha\beta}\,\nabla_{\beta}\psi)
        -\frac{1}{20}\psi\,\Box K^\alpha{}_\alpha
        \nonumber\\
        &+\left(\frac{3}{10}R_{\alpha\beta}\,K^{\alpha\beta}+\frac{1}{30}R\, K^\alpha{}_\alpha+f\right)\psi\,,
\end{align}
where the {scalar function $f$} must satisfy the following geometric condition\cite{Michel_2014}:
\begin{equation}\label{eq:Obs}
\nabla_\alpha f=-\frac{2}{15}\Bigl(C^\beta{}_{\gamma\delta\alpha}\nabla_{\beta}K^{\gamma\delta}-3K^{\gamma\delta}\nabla_\beta C^\beta{}_{\gamma\delta\alpha} \Bigr)\,,
\end{equation}
where $C^\alpha{}_{\beta\gamma\delta}$ is the Weyl tensor. 
This is met for the off-shell Carter metric and the explicit expression can be found in \cite{Gray:2021wzf}. 
If $f=0$ then the symmetry operators can be written in terms of the conformal Killing vectors and conformal Killing tensors however, because this is not true for the Carter metric, the full Killing tensors are required for such operators and the separability follows by transforming the conformal wave equation into the Carter spacetime by rescaling by the conformal factor $\Omega$ -- this provides the function ${\cal R}(x)=\Omega(x)$.

We demonstrate this explicitly below. 
Namely, if one substitutes  an ansatz of the form
\begin{align}
    \psi=\Omega\, \psi_0\, e^{\im (\omega \,\eta+n\,\sigma)}\,.
\end{align}
where $\psi_0=\psi_0(p,q)$, {and $\omega$ and $n$ are constants, not necessarily related to the physical frequencies and angular momenta (this follows from the fact that $\eta$ and $\sigma$ are not necessarily the proper time and angular coordinates, as discussed in \cite{Ovcharenko:2025cpm})}. 
Then \eqref{cKG_eq} becomes
\begin{align}
    (Q \psi_0')'&+(P \dot{\psi_0})^{\dot{}}+\dfrac{\psi_0}{6}(Q''+\ddot{P})+\nonumber\\
    &+\psi_0\dfrac{(\omega q^2-n)^2}{Q}-\psi_0 \dfrac{(\omega p^2+n)^2}{P}=0\,.\label{KG_simpl}
\end{align}

Obviously, this equation is separable, and by substitution $\psi_0=R_0(q)X_0(p)$, 
one obtains
\begin{align}
    &(P \dot{X}_0)^{\dot{}}\,+\Big[\,\dfrac{\ddot{P}}{6}\,-\dfrac{(\omega p^2+n)^2}{P}+\lambda\Big]X_0=0\,,
    \nonumber
    \\
    &(Q R'_0)'+\Big[\dfrac{Q''}{6}+\dfrac{(\omega q^2-n)^2}{Q}-\lambda\Big]R_0=0\,,\label{rad_teuk_0}
\end{align}
where, as before,  ${'= \frac{d}{dq}}$ and ${\,\dot{} =\frac{d}{dp}}$.

{These equations are already in the separated Teukolsky-like form. 
However, for computational purposes, a description of the matter field using the effective potential is sometimes required (for example, for determination of the quasinormal modes using the WKB method, see \cite{Konoplya2011}). 
Below we provide a construction of the effective potential for eqs. (\ref{rad_teuk_0}).} 

Namely, if we introduce new coordinates 
\begin{align}
    \dd \bar{q}=\dfrac{\dd q}{Q}\,,\quad \dd \bar{p}=\dfrac{\dd p}{P}\,,
\end{align}
these equations become
\begin{align}
    f_{,\bar{q}\bar{q}}+f\, V_Q=0\,,\quad
    f_{,\bar{p}\bar{p}}+f\,V_P=0\,,
\end{align}
where
\begin{align}
    V_Q&=\ \ (\omega q^2-n)^2-\lambda Q+\dfrac{1}{6}\Big(\dfrac{Q_{\bar{,q}}}{Q}\Big)_{,\bar{q}}\,,\nonumber\\
    V_P&=-(\omega p^2+n)^2+\lambda P+\dfrac{1}{6}\Big(\dfrac{P_{\bar{,p}}}{P}\Big)_{,\bar{p}}\,.
\end{align}
These are two decoupled ODEs which may be solved numerically for the given metric of choice of the conformal class.
We now turn to the next lowest spin field, namely spin-1/2.

\subsection{Massless Dirac equation}
Let us now turn to a spin 1/2 particle, described by the Dirac equation. Since \eqref{off-shell-PD} has a non-constant conformal factor, it appears that the massive Dirac equation is not separable. 
For this reason we turn to the {\em massless} Dirac equation, which reads  
\be \label{masslessDirac}
\gamma^a\,\nabla_a \psi=0\,. 
\ee 
Here, $\gamma^a$ are the gamma matrices, obeying $\{\gamma^a,\gamma^b\}=2g^{ab}$, and $\nabla_a$ stands for the spinorial derivative  
\be 
\nabla_a=\partial_a+\frac{1}{4}\omega_{abc}\gamma^b\gamma^c\,,
\ee 
where $\partial_a=e_a\cdot \partial$ is a derivative in the direction $e_a$, $\omega_{abc}$ are the standard spin coefficients with respect to tetrad frame $e^a$, and the 1-forms of the curvature $\omega^b{}_c=e^a\omega_a{}^b{}_c$ obey Cartan's equation ${\dd e^a+\omega^a{}_b \wedge e^b = 0}$.

Importantly, the massless Dirac equation \eqref{masslessDirac} is Weyl invariant. 
More specifically, the following conformal transformation
\be 
g\to \frac{1}{\Omega^2}\,g\,,\quad \psi \to \frac{1}{\sqrt{\Omega}}\,\psi\,,
\ee 
leaves it unchanged. 
This means that instead of considering the massive Dirac equation in the spacetime \eqref{off-shell-PD}, we can focus on the massless Dirac equation in the of-shell Carter spacetime \eqref{Carter}. 
However, such an equation is known to be ${\cal R}$-separable, with the separation constant given by the eigenvalue of the symmetry operator (commuted with the Dirac operator) constructed from the PKY tensor of the Carter spacetime, e.g. \cite{McLenaghan1979, Frolov:2017kze}.

Namely, following {the notation of} \cite{Collas2019}, let us choose the following frame for the Carter spacetime:
\begin{align}\label{orthogonal-tetrad}
    {e}_0&=\dfrac{1}{\sqrt{Q \rho^2}}\,(q^2\partial_{\eta}-\partial_{\sigma})\,,\quad
    {e}_1 =\sqrt{\dfrac{ Q}{\rho^2}}\,\partial_q\,,\nonumber\\
    {e}_2&=\sqrt{\dfrac{P}{\rho^2}}\,\partial_p\,,\quad
    {e}_3 =\dfrac{1}{\sqrt{P \rho^2}}\,(p^2\partial_{\eta}+\partial_{\sigma})\,,
\end{align}
together with the Weyl representation for the gamma matrices,
\begin{align}
    \gamma^a=\begin{pmatrix}
        0 & \sigma^a\\
        \bar{\sigma}^a & 0
    \end{pmatrix}\, ,\qquad\sigma^a=(1,\sigma^k)\,,\qquad\bar{\sigma}^a=(1,-\sigma^k)\,,
\end{align}
where $\sigma^k$ are the Pauli matrices. 
Then we find that the massless Dirac equation in the conformal-to-Carter spacetime \eqref{off-shell-PD} with the following ansatz for the spinor ${\psi=\frac{1}{\sqrt{\Omega}}\,\psi^\C}$,
\begin{align}
    \psi=\dfrac{1}{2\sqrt{\Omega}}\begin{pmatrix}
        \sqrt{q+\im\, p}\begin{pmatrix}
            \psi_-^L-\psi_{+}^L\\
        \psi_-^L+\psi_{+}^L
        \end{pmatrix}\\
        \,\\[-5mm]
        \sqrt{q-\im\, p}\begin{pmatrix}\psi_{+}^{R}-\psi_{-}^{R}\\
        \psi_{+}^R+\psi_{-}^R\end{pmatrix}
    \end{pmatrix}e^{\im (\omega \eta+n \sigma)}\,,\label{psi_ansatz}
\end{align}
is separable, namely
\begin{align}
    \psi_+^\sigma=R_+^\sigma X_+^\sigma\,,\qquad \psi_-^\sigma=R_-^\sigma X_-^\sigma\,,
\end{align}
(here, and further on, the upper index $\sigma$ stands for ``$L$'' or ``$R$''). 
Specifically, these radial and angular functions obey the following set of ODEs:
\begin{align}
    &\sqrt{Q}\,D_+^q R_+^\sigma=\lambda_1 R_-^\sigma\,,\quad 
    \sqrt{Q}\,D_-^qR_-^{\sigma}=\lambda_2 R_+^{\sigma}\,,\nonumber\\
    &\im\sqrt{P}\,D_-^p X_-^{\sigma}=\lambda_1 X_+^{\sigma}\,,\quad 
    \im\sqrt{P}\,D_+^p X_+^{\sigma}=\lambda_2 X_-^\sigma\,.
\end{align}
In these expressions, $\lambda_{1}$ and $\lambda_{2}$ are two independent {\em separation constants}, and we have introduced the following differential operators: 
\begin{align}\label{DiracODEs1}
    D_{\pm}^q=\partial_q\pm \im\, \dfrac{q^2\omega-n}{Q},~~~D_{\pm}^p=\partial_p\pm \sigma\,\dfrac{p^2\omega+n}{P}\,,
\end{align}
where $\sigma=-1$ if the spinor is left, while $\sigma=1$ if the spinor is right. 

By combining \eqref{DiracODEs1}, we can obtain the following decoupled equations:
\ba
\sqrt{P}\,D_-^p(\sqrt{P}\,D_+^p X_+^{\sigma})&=&-\lambda X_+^{\sigma}\,,\nonumber\\
\sqrt{P}\,D_+^p(\sqrt{P}\,D_-^p X_-^{\sigma})&=&-\lambda X_-^{\sigma}\,,\nonumber\\
\sqrt{Q}\,D_-^q(\sqrt{Q}\,D_+^q R_+^{\sigma})&=&\lambda R_+^{\sigma}\,,\nonumber\\
\sqrt{Q}\,D_+^q(\sqrt{Q}\,D_-^q R_-^{\sigma})&=&\lambda R_-^{\sigma}\,,
\ea
where we have denoted $\lambda=\lambda_1\lambda_2$. 
Notice that the two separation constants $\lambda_1$ and $\lambda_2$ have been combined into one, so that the resulting separated equations feature only one separation constant --- like in the massive case~\cite{Chandrasekhar1998,Frolov:2017bdq}.
Moreover, upon the following substitution: $X_{\pm}\to P^{1/4}X_{\pm}$,  $R_{\pm}\to \sqrt{Q}\,R_{\pm}$, one obtains that these equations become
\ba
(P\dot{X}_{\pm})^{\dot{}}
+\!\Big[\dfrac{\ddot{P}}{4}\pm 2\sigma\omega\, p-\dfrac{(n+\omega p^2\pm \frac{\sigma}{4}\dot{P})^2}{P}+\lambda\Big]X_{\pm}&=&0\,,\nonumber\\
 \frac{1}{\sqrt{Q}}(Q^\frac{3}{2}R'_{\pm})'+\!
    \Big[\dfrac{Q''}{2}\pm 2 \im \omega\, q+K\dfrac{K\mp\frac{\im}{2} Q'}{Q}-\lambda\Big]R_{\pm}&=&0\,,\nonumber\\
    &&\label{rad_teuk_1/2}
\ea
where ${K=\omega q^2-n}$.

The equations (\ref{rad_teuk_1/2}) can be used to find $\lambda$ and $\omega$. However, the radial equation is complex, and for practical reasons, one may be interested in finding a real effective potential (see \cite{Chandrasekhar1998}). 
For this, let us analyze the radial equations that are given by
\ba
    \sqrt{Q}\,\Big(\dfrac{\dd}{\dd q}+\im\, \dfrac{q^2 \omega-n}{Q}\Big)R_+^{\sigma}&=&\lambda_1 R_-^{\sigma}\,,\nonumber\\
    \sqrt{Q}\,\Big(\dfrac{\dd}{\dd q}-\im\, \dfrac{q^2 \omega-n}{Q}\Big)R_-^{\sigma}&=&\lambda_2 R_+^{\sigma}\,.
\ea

If we introduce a new coordinate $q_*$ such that
\begin{align}
    \dfrac{\dd}{\dd q}=\dfrac{q^2-n/\omega}{Q}\dfrac{\dd}{\dd q_*}\,,
\end{align}
and the new functions:
\begin{align}
    Z_{\pm}^{\sigma}=\sqrt{\frac{\lambda_2}{\lambda}}R_+^{\sigma} \mp \sqrt{\frac{\lambda_1}{\lambda}}R_-^{\sigma}\,,
\end{align}
then the corresponding equations become 
\begin{align}
    \Big(\dfrac{\dd}{\dd q_*}\pm \sqrt{\lambda} \dfrac{\sqrt{Q}}{q^2-n/\omega}\Big)Z_{\pm}^{\sigma}=-\im\, \epsilon Z_{\mp}^{\sigma}
\end{align}
By combining these equations, one then obtains that the functions $Z_{\pm}$ can be described by the following equations:
\begin{align}
    \Big(\dfrac{\dd^2}{\dd q_*^2}+\omega^2\Big)Z_{\pm}=V_{\pm}Z_{\pm}\,,
\end{align}
where
\begin{align}
    V_{\pm}=\lambda \dfrac{Q}{(q^2-n/\omega)^2}\pm \sqrt{\lambda Q}\,\dfrac{\dd}{\dd q_*}\Big(\dfrac{\sqrt{Q}}{q^2-n/\omega}\Big)\,.
\end{align}
Analogous construction can also be done for the angular equation; however, we do not provide an explicit calculation for it.

\subsection{Maxwell equation}\label{sec_Maxwell}

Next on our spin ladder spins are the test ``Maxwell equations'':%
\footnote{Note that similar to the charged Pleba\'nski--Demia\'nski case, here we have to assume that the massless vector field $B_\mu$ is different from the underlying Maxwell field $A_\mu$. 
Otherwise, one has to deal with a coupled system of Einstein--Maxwell equations, which likely does not separate. 
}
\be 
\nabla_\mu H^{\mu\nu}=0\,,\quad H_{\mu\nu}=2\,\nabla_{[\mu} B_{\nu]}\,.
\ee

As is well known, such equations are also conformally invariant, namely, the following transformations
\be 
g\to \Omega^2 g\,,\quad B_\mu\to B_\mu\,,
\ee 
leaves them invariant. 
This means that instead of studying the Maxwell equations on the spacetime \eqref{off-shell-PD}, we can equally study them on the Weyl rescaled spacetime \eqref{Carter}. 
Since this is simply Carter's off-shell spacetime we can employ the technique of separation {developed by Lunin, and Frolov--Krtou{\v s}--Kubiz\v{n}\'ak (LFKK)  \cite{Lunin:2017drx, Frolov:2018pys}. 
Such an approach is distinct} from the Teukolsky separation \cite{Teukolsky1973,TorresdelCastillo1988}. 
The analogous separation can be also done for the off-shell Pleba\'nski--Demia\'nski spacetime (for the result see Sec. \ref{sec_Teuk_eq} with $s=1$), but in this subsection we choose a distinct approach that directly uses the CKY tensors, constructed previously.  

Namely, we can make the following ansatz for the vector potential:
\be 
B^\alpha=P^{\alpha\beta}\,\nabla_\beta Z\,,
\ee 
where $P^{\alpha\beta}$ is the ``polarization tensor'' obeying 
\be 
(g_{\alpha\beta}+\im\,\mu \, h_{\alpha\beta})P^{\beta\gamma}=\delta^\gamma_\alpha\,,
\ee 
$\mu$ is a separation constant, and $ h_{\alpha\beta}$ is the corresponding principal Killing--Yano tensor of the Carter off-shell spacetime. 
We can then set 
\be 
Z=R_1(q)X_1(p)\,{\rm e}^{\im (\omega \eta + n \sigma)}\,.
\ee 
and following the steps in \cite{Frolov:2018pys}, which importantly involves imposing the Lorenz gauge {in the Carter spacetime}%
\footnote{Note that the Lorenz gauge condition is not conformally invariant. 
That is why although the resulting vector potential is a solution of Maxwell equations in the conformally related spacetime it no longer obeys the Lorenz gauge condition, but rather
\be
\nabla'_\mu B'^\mu=-2\,B'{}^\mu\,\nabla_\mu\log\Omega\,.
\ee
},
\be 
\nabla_\mu B^\mu=0\,,
\ee 
to show that the Maxwell equations are satisfied, provided the following separated equations are satisfied:
\ba 
\Bigl(\frac{P \dot X_1 }{z_p}\Bigr)^{\dot{}}-\Big[\frac{(\omega z_p+\nu)^2}{\mu^4 P z_p}+\frac{2-z_p}{\mu z_p^2}\nu\Big]X_1&=&0
\,,\nonumber\\
\Bigl(\frac{Q R_1'}{z_q}\Bigr)'+\Big[\frac{(\omega z_q+\nu)^2}{\mu^4 Q z_q}+\frac{2-z_q}{\mu z_q^2}\nu\Big]\,R_1&=&0\,.
\ea 
Here,
\be 
\nu=-\mu^2n-\omega\,,\quad z_q=1+\mu^2q^2\,,\quad z_p=1-\mu^2p^2\,.
\ee 
Moreover, it can be shown \cite{Cynthia} that the above separability can be linked to the existence of recently constructed symmetry operators \cite{Gray:2024kad}.  

\subsection{Gravitational perturbations}

{The last field whose separability we would like to investigate is the gravitational perturbations. 
In our further analysis, we will assume the corresponding source terms can be omitted. 
While this assumption may seem unnatural, it was used to obtain the quasinormal modes of Kerr--Newman black holes (see \cite{Kokkotas1993,Berti2005}). 
Furthermore, as was obtained in \cite{saha2025}, variation of the matter field for the Kerr--Newman black hole does not change the frequencies of the quasinormal modes significantly and becomes important only near extremality. 
Thus, keeping in mind that this ``freezing" assumption is not fully correct for a coupled electromagnetic and gravitational field, we will use it throughout this subsection as a reasonable approximation. 
Moreover, the coupled set of equations may not be separable, see \cite{Dias2022,Dias2015}.}

There are various ways to describe pure gravitational perturbations. 
One consists of considering perturbations of the metric ${g_{\mu\nu}=g_{\mu\nu}^{(0)}+\epsilon\,g_{\mu\nu}^{(1)}+{\cal O}(\epsilon)}$; by expanding field equations in $\epsilon$ one obtains linear equations for the perturbations. 
The problem for our work is that the resulting equations are not conformally invariant  and, moreover, are gauge dependent.

Another way is to work in the Newman--Penrose formalism and consider perturbations not of the metric, but of the Weyl and optical scalars, namely assuming ${\Psi_i=\Psi_i^{(0)}+\epsilon\, \Psi_i^{(1)}+{\cal O}(\epsilon)}$. 
We will use this method because for rotating metrics it was shown to be much simpler then the first one. 

{Let us consider the background metric to be the off-shell conformal-to-Carter metric \eqref{off-shell-PD}.
This has several important algebraic properties and the corresponding field equations thus simplify significantly. 
In particular, it is of type D, and the corresponding PNDs are geodesic and shear-free, see Eqs.~(7) and~(13) in~\cite{Ovcharenko:2025cpm},
\begin{align}
    \Psi_0^{(0)}=\Psi_1^{(0)}=0=&\ \Psi_3^{(0)}=\Psi_4^{(0)}\,,
    \nonumber\\
    \kappa^{(0)}=\nu^{(0)}=0=&\ \sigma^{(0)}=\lambda^{(0)}\,.
\end{align}
}
Because such conditions are quite strong --- and are the same as for the the on-shell PD spacetime ---  one can conveniently employ the results for the separability of the gravitational perturbations in the PD spacetime obtained previously by Dudley and Finley \cite{dudley1977}. 

Interestingly, to obtain the separable equations, instead of directly using the scalars $\Psi_0^{(1)}$ and $\Psi_4^{(1)}$ it is necessary to introduce new functions $\psi_{\pm 2}$ such that
\begin{align}
 \Psi_0^{(1)}&={\rm e}^{\im (\omega \eta + n \sigma)}\dfrac{\Omega^3}{(q+\im p)^2}\,\psi_{-2}\,,   \nonumber\\
 \Psi_4^{(1)}&={\rm e}^{\im (\omega \eta + n \sigma)}\dfrac{\Omega^3}{(q+\im p)^2}\,\psi_2\,.
\end{align}
Then the equations for gravitation perturbations become independent of the conformal factor, separable, and by making a factorization ${\psi_{\pm2}=Q\,R_{\pm2}\,X_{\pm2}}$, one gets
\begin{align}
(P\dot{X}_{\pm2})^{\dot{}}+\Big[\dfrac{3}{2}\ddot{P}\pm 8\omega \,p-\dfrac{(\omega\, p^2+n\pm \dot{P})^2}{P}+\lambda\Big]X_{\pm 2}&=0\nonumber\,,\\
\dfrac{(Q^3R_{\pm 2}')'}{Q^2}+
    \Big[\dfrac{5}{2}Q''\pm 8 \im \omega\,q+K\,\dfrac{K\mp2\im\, Q'}{Q}-\lambda\Big]R_{\pm2}&=0\,,\label{rad_teuk_2}
\end{align}
where $\lambda$ is the separation constant and ${K=\omega\, q^2-n}$ (to obtain these equations, we have used a frame presented in Eq.~(6) in \cite{Ovcharenko:2025cpm}).

The important observation here is that the {\em equation for gravitational perturbations is conformally invariant} even though generally gravitational perturbations do not obey such a property (unlike the conformal Klein--Gordon or massless Dirac equation). 
This specific property is related to the special algebraic properties of the metric \eqref{PD}, namely that the spacetime is of algebraic type D, and the corresponding PNDs are geodesic and shear-free. See ~\cite{Araneda:2018ezs} for a discussion of the underlying conformal invariance of the Teukolsky equation.

In summary, the equations \eqref{rad_teuk_2} hold for the off-shell {conformal-to-Carter metric \eqref{off-shell-PD}.} 
Although this observation is quite straightforward from the construction in \cite{dudley1977}, previously there were no other solutions within the off-shell PD class other than the on-shell PD solution. 
Currently, given that the new class of OP metrics has been found, this observation gains its importance, as the results of \cite{dudley1977} can be directly translated to this case (with the properly specified on-shell functions $Q$, $P$ and $\Omega$, see Appendix~\ref{App:onshell}).

\subsection{Teukolsky-like equation}\label{sec_Teuk_eq}

In the previous sections, we have analyzed field equations for various fields with various spins. 
It can be seen that all the resulting separate equations are very similar, and hence we can combine them into two simple general equations. 
Namely, the angular equations \eqref{rad_teuk_0}, \eqref{rad_teuk_1/2} and \eqref{rad_teuk_2} can be unified into the equation
\begin{widetext}
\begin{align}
    \Bigg[\dfrac{\dd}{\dd p}\Big(P \dfrac{\dd}{\dd p}\Big)+\dfrac{1}{6}(2s^2+1)\ddot{P}\pm 4 s\, \omega\, p-\dfrac{(\omega \,p^2+n\pm\frac{s}{2}\dot{P})^2}{P}+\lambda\Bigg]X_{\pm s}=0\,,\label{pol_teuk_gen}
\end{align}
where the plus and minus sign represent the right-handed and left-handed polarizations respectively, and $s$ is the spin of the particle. 
That is, ${s=0}$ corresponds to the Klein--Gordon field, ${s=1/2}$ to the massless Dirac field, and ${s=2}$ to the gravitational perturbations. 
Similarly, the radial equations \eqref{rad_teuk_0}, \eqref{rad_teuk_1/2} and \eqref{rad_teuk_2} can be combined into
\begin{align}
    \Bigg[Q^{-s}\dfrac{\dd}{\dd q}\Big(Q^{s+1}\dfrac{\dd}{\dd q}\Big)+\dfrac{1}{6}(2s+1)(s+1)\,Q''\pm 4\im s\, \omega\, q+K\,\dfrac{K\mp\im s\, Q'}{Q}-\lambda\Bigg]R_{\pm s}=0\,,\label{rad_teuk_gen}
\end{align}
where ${K=\omega\,q^2-n}$.
\end{widetext}

Note that the analogous equations were obtained in \cite{dudley1977}, where the authors considered the PD spacetime. 
Our novel result is that the same equation also holds for the off-shell {conformal-to-Carter} spacetime \eqref{off-shell-PD}. 

Notice also that the set of equations \eqref{pol_teuk_gen}--\eqref{rad_teuk_gen} was obtained by generalizing the equations for ${s=0}$, ${s=1/2}$, ${s=2}$. 
The case of $s=1$, considered in Sec \ref{sec_Maxwell}, does not fit into this result because we have used a different approach, namely the LFKK ansatz directly on the field. 
However, the Teukolsky-like approach can be also applied to yield the general set of equations \eqref{pol_teuk_gen}--\eqref{rad_teuk_gen} for ${s=1}$.

\section{Non-twisting case}\label{Sec:Nontwisting}

Although the OP spacetime in the form \eqref{OP} and the more general geometry \eqref{off-shell-PD} represent quite simple metrics that describes a large family of type~D exact solutions, including black holes, it does not allow for directly obtaining the {\em non-twisting} limit. 
This is because the optical scalars, related to the principal null directions (PNDs) of the Weyl tensor (namely $\rho_{sc}$ and $\mu_{sc}$), always have a non-zero imaginary part (see Eq.~(7) in \cite{Ovcharenko:2025cpm}), namely
\begin{align}
    \rho_{sc}=\mu_{sc}=\dfrac{1}{2}\sqrt{\dfrac{Q\Omega^2}{2\rho^2}}\Big[\Big(\ln\dfrac{\Omega^2}{\rho^2}\Big)_{,q}+\im (\ln \rho^2)_{,p}\Big]\,.
\end{align}
As the imaginary part is related to the twist of the corresponding two null PND congruences, this means that the metric \eqref{off-shell-PD}  describes only twisting geometries. 
However, physically, non-twisting spacetimes attract strong attention (including various extensions of the famous Schwarzschild black hole), and it is thus worth considering them also here.

\subsection{No-twisting metric and its hidden symmetries}

There are several ways to consider the non-twisting subcase. 
One approach is to perform a set of transformations of coordinates and parameters, introducing the twist parameter explicitly and then taking the zero-twist limit. 
This was explicitly performed in Sec.~V of \cite{Ovcharenko:2025cpm}.
However, such procedure would be unnecessarily complicated for the construction we have presented above. 
We thus adopt here an alternative approach, namely to start {\em directly} with the {\em non-twisting version of the off-shell metric} \eqref{off-shell-PD}, which reads
\begin{align}
    \dd s^2=\dfrac{1}{\Omega^2}\Big[-\mathcal{Q}\,\dd \eta^2+\dfrac{\dd q^2}{\mathcal{Q}}+q^2\Big(\dfrac{\dd p^2}{P}+P\,\dd\sigma^2\Big)\Big]\,.\label{metr_non_twist}
\end{align}

Employing the field equations, together with specifying the corresponding vector potential, one can obtain various non-twisting solutions, for example the Schwarzschild--BR (see Sec. \ref{sec_Schw_BR}) or the non-twisting C-metric (see  Sec. \ref{sec: PD}). 
To keep our results here as general as possible, we will not stick to any explicit spacetime, and we will investigate symmetry properties of the off--shell geometry \eqref{metr_non_twist}. 

As before, in what follows, we shall thus consider the {\em general off-shell spacetime} of the form \eqref{metr_non_twist} with
\be 
\mathcal{Q}=\mathcal{Q}(q)\,,\quad P=P(p)\,, \quad \Omega=\Omega(p,q)\,.
\ee 
One can easily check that in this case the optical scalars $\rho_{sc}$ and $\mu_{sc}$ have zero imaginary part, which means that the class of metrics \eqref{metr_non_twist} is indeed non-twisting.

Such a class of geometries is clearly equipped with two Killing vectors, $\partial_\eta$ and $\partial_\sigma$. 
In addition, it admits a CKY 2-form $h$  (this time degenerate) given by 
\be\label{hdeg} 
h=\dfrac{q}{\Omega^3}\,\dd\eta\wedge\dd q\,,
\ee 
i.e.
\be\label{hdeg-b} 
h=\frac{1}{\Omega^3}\,\dd b\,,\quad \hbox{where}\quad 
2b=-q^2\dd \eta\,.
\ee 
Despite being degenerate, this tensor remains non-trivial, and provides an additional independent symmetry allowing for separability of various test field equations, see below.
Its dual ${k=\star h}$ is also a CKY 2-form that is explicitly given by
\begin{align}
    k=\dfrac{q^3}{\Omega^3}\,\dd p\wedge \dd \sigma\,.
\end{align}
The corresponding conformal Killing tensors are
\begin{align}
    Q_{(h)}&=\dfrac{q^2}{\Omega^4}\Big(\mathcal{Q}\,\dd \eta^2-\dfrac{\dd q^2}{\mathcal{Q}}\Big)\,,\\
    Q_{(k)}&=\dfrac{q^4}{\Omega^4}\Big( P\,\dd \sigma^2 + \dfrac{\dd p^2}{P}\Big)\,.
\end{align}

Similar to the twisting case, we may ask ourselves for which conformal factors $\Omega$ one can deduce the existence of the two Killing vectors, namely $\partial_\eta$ and $\partial_\sigma$, from the above hidden symmetry. 
The procedure of investigation essentially follows that in Sec.~\ref{Sec:Tower}, and {and the results are} summarized in Table~\ref{Table:2}. 
It follows from this table that there are two interesting conformal factors, namely the PD conformal factor  
\be 
\Omega=1-pq\,,
\ee 
and the novel conformal factor 
\be \label{Noveldeg}
\Omega=\frac{q^2}{1-pq}\,,
\ee 
for which we recover both isometries of this spacetime from its hidden symmetry.%
\footnote{Here, for convenience, we have ``normalized'' the Killing vector constants $a$ and $b$ to $\pm 1$.} 
The latter seems quite interesting and worth investigating in the future.

\renewcommand{\arraystretch}{1.5}

\begin{center}
\setlength{\tabcolsep}{6pt}
\begin{table*}[t!]

\begin{tabular}{ |c||c|c|c|p{2.5cm}| }
 \hline
 \multicolumn{5}{|c|}{$\text{Non-twisting case:}\  \text{Killing vector of the form}\quad a\,\partial_\eta+b\,\partial_\sigma$ } \\
 \hline
 Vector & ${\,a\neq0\neq b\,}$  &$a=0\,, b\neq 0$& $a\neq0\,, b=0$ & ${\quad a=0=b}$\\
 \hline
 \hline
 $\xi=\frac{1}{3}\nabla\cdot h$  & \xmark    &\xmark&   $\Omega=a+q\,F(p)$& $\Omega=q\,F(p)$ \\
 $\zeta=\frac{1}{3}\nabla\cdot k$  & \xmark    & $\Omega=F(q)+b\,pq$&  \xmark  & $\Omega=F(q)$  \\
 $\eta=K_k\cdot \xi$  & \xmark    & \xmark &   \xmark & $\Omega(p,q)$\ \text{arbitrary} \\
  $\hat{\eta}=K_h\cdot \zeta$  & \xmark    &\xmark&   \xmark& $\Omega(p,q)$\ \text{arbitrary} \\
$\iota=K_k\cdot \zeta$  & \xmark    &   $\displaystyle\Omega=\frac{q}{q\,F(q)-b\,p}$ & \xmark & $\Omega=F(q)$\\
$\hat{\iota}=K_h\cdot \xi$  & \xmark    &\xmark& $\displaystyle\Omega=\frac{q^2}{a+q\,F(p)}$ & $\Omega=q\,F(p)$ \\
 \hline
\end{tabular}
\caption{
{\bf Construction of the Killing vectors for the non-twisting case.} 
This table summarizes various conformal factors $\Omega$, and the associated spacetimes $g$, \eqref{metr_non_twist}, that give rise to Killing vector fields from the existence of a degenerate hidden symmetry of the CKY 2-form $h$, \eqref{hdeg}, and its dual  $k=\star h$.
Here, $a$ and $b$ are the constants determining the specific Killing vector of the form ${\,a\,\partial_\eta+b\,\partial_\sigma\,}$, and $F(\,\cdot\,)$ is an arbitrary function of it argument. 
As before,  the symbol \xmark\  indicates that the solution does not exist.  
Since the non-twisting OP spacetime \eqref{metr_non_twist} is not in this table, this proves that although it admits two Killing vectors ${\partial_\eta}$ and ${\partial_\sigma}$, these cannot be generated (in a straightforward manner, at least) from their two CKY 2-forms \eqref{CKY-OP-h} and \eqref{CKY-OP-k}. 
}
\label{Table:2}
\end{table*}
\end{center}

\subsection{Separability properties}

Now, let us also briefly discuss how the test field or particle  equations simplify for the non-twisting case. 
{The Hamilton--Jacobi equation (\ref{HJ_eq_2})  in this case becomes
\begin{align}
    \dfrac{L^2}{P}-\dfrac{E^2q^2}{\mathcal{Q}}+q^2\mathcal{Q}\,(W')^2+P\,(\dot{W})^2=0\,.
\end{align}
Contraction of the conformal Killing tensors with the 4-momentum gives an analogue of the equation \eqref{Cart_const},
\begin{align}
    q^2\Big(\dfrac{E^2}{\mathcal{Q}}-\mathcal{Q}\,(W')^2\Big)=C\,.
\end{align}
Thus, the equation of geodesic motion becomes separable, and we obtain
\begin{align}
    P\,(\dot{W})^2=C-\dfrac{L^2}{P}\,,\quad
    q^2\mathcal{Q}\,(W')^2=\dfrac{E^2q^2}{\mathcal{Q}}-C\,.
\end{align}
The corresponding 4-velocities are thus given by
\begin{align}
    \frac{\dd\eta}{\dd\lambda}&=\dfrac{\Omega^2}{\mathcal{Q}}E\,,\quad \frac{\dd q}{\dd\lambda}=\dfrac{\Omega^2}{q}\sqrt{E^2q^2-C\mathcal{Q}}\,,\nonumber\\
    \frac{\dd\sigma}{\dd\lambda}&=\dfrac{\Omega^2}{P}\dfrac{L}{q^2}\,,\quad \frac{\dd p}{\dd\lambda}=\dfrac{\Omega^2}{q^2}\sqrt{CP-L^2}\,.
\end{align}

It can also be shown that the analogues of the Teukolsky-like equations (\ref{pol_teuk_gen})--(\ref{rad_teuk_gen}) are 
\begin{align}
    &\Bigg[\dfrac{\dd}{\dd p}\Big(P \dfrac{\dd}{\dd p}\Big)+\dfrac{1}{6}(2s^2+1)\ddot{P}-\dfrac{(n\pm\frac{s}{2}\dot{P})^2}{P}+\lambda\Bigg]X_{\pm s}=0\,,\\
    &\Bigg[Q^{-s}\dfrac{\dd}{\dd q}\Big(Q^{s+1}\dfrac{\dd}{\dd q}\Big)+\dfrac{1}{6}(2s+1)(s+1)\,Q''\nonumber\\
    &\hspace{28mm}\pm 4\im s\, \omega\, q+K\,\dfrac{K\mp\im s\, Q'}{Q}-\lambda\Bigg]R_{\pm s}=0\,,
\end{align}
where now ${K=\omega\, q^2}$, $Q=q^2 \mathcal{Q}$.

\section{Two new special cases}
\label{Sec:SpecialCases}
Finally, it is worth discussing explicitly two important subcases of the OP spacetimes, namely the Kerr-BR and the Schwarzschild--BR metrics, which correspond to (rotating) black holes in equilibrium with a background magnetic field. 
These are {\em exact solutions} to the Einstein--Maxwell equations. 
We shall provide explicit forms of their hidden symmetries, as well as calculate the Penrose charges mentioned above, providing new physically motivated examples of these non-trivial charges.

\subsection{Kerr--BR black holes}

\subsubsection{Solution}
The Kerr--BR black hole metric \cite{Podolsky:2025tle, Ovcharenko:2025cpm} has the form 
\begin{align}
    \dd s^2=&\frac{1}{\Omega^2}\Bigl[-\frac{Q}{\rho^2}(\dd t-a\sin^2\theta\, \dd\varphi)^2+\frac{\rho^2}{Q}\dd r^2+\frac{\rho^2}{\tilde{P}}\dd\theta^2\nonumber\\
    &\qquad+\frac{\tilde{P}}{\rho^2}\sin^2\theta\bigl(a \,\dd t-(r^2+a^2) \dd\varphi\bigr)^2\Bigr]\,,\label{metr_Kerr--BR}
\end{align}
where ${\rho^2=r^2+a^2\cos^2\theta}$.
The metric functions are explicitly
\begin{align}
\tilde{P} & = 1 + B^2 \Big(m^2\,\dfrac{I_2}{I_1^2} - a^2 \Big)\cos^2\theta\,,  \nonumber\\
{Q} & = \,\big(1+B^2r^2\big)\, \Delta\,,\nonumber\\[2mm]
\Omega^2 & = \,\big(1+B^2r^2\big) - B^2 \Delta \cos^2\theta\,,\label{Omega}
\end{align}
with
\begin{align}
\Delta &= \Big(1-B^2m^2\,\dfrac{I_2}{I_1^2}\Big) r^2-2m\,\dfrac{I_2}{I_1}\,r + a^2
\,,\nonumber\\
    I_1 &= 1-\tfrac{1}{2} B^2a^2  \,,\qquad
    I_2 = 1-B^2a^2   \,.\label{I1I2}
\end{align}

As discussed in \cite{Podolsky:2025tle}, the axes of \eqref{metr_Kerr--BR} are not regular unless the periodicity of the angle $\varphi$ is $[0,2\pi C)$, where $C$ is given by
\begin{equation}
    C=\frac{1}{\tilde P(0)}=\frac{1}{\tilde P(\pi)}=\left[ 1+B^2\left(m^2\frac{I_2}{I_1^2}-a^2\right)\right]^{-1}\,.
\end{equation}

The corresponding 4-potential, solving the Einstein--Maxwell field equations for the metric \eqref{metr_Kerr--BR}, can be obtained by taking the real counterpart of Eq.~(2.8) in \cite{Podolsky:2025tle} and is given by
\begin{align}
A=&\dfrac{1}{B\,(r^2+a^2 \cos^2\theta)}\times\nonumber\\
  &\Big[\ \Omega_{,r}\big(r\cos\gamma+a \cos\theta \sin\gamma\big)\big(a\,\dd t-(r^2+a^2)\dd \varphi\big)\nonumber\\&
  +\dfrac{\Omega_{,\theta}}{\sin\theta}\big(a\cos\theta \cos \gamma-r \sin \gamma\big)\big(\dd t-a \sin^2\theta\, \dd \varphi\big)\Big]\nonumber\\
  &+\cos\gamma\, \dfrac{\Omega-1}{B}\,\dd \varphi\,,\label{4pot_Kerr-BR}
\end{align}
where $\gamma$ is the duality rotation parameter (mixing the electric and magnetic components of the Maxwell field)} and $\Omega$ is the conformal factor \eqref{Omega}.

The Kerr--BR metric (\ref{metr_Kerr--BR}) is obtained from \eqref{OP} by the specific coordinate transformation, namely
\begin{align}
    t=\eta-a^2\sigma\,,\quad
    \varphi=-a\sigma\,,\quad
    a\cos\theta=p\,,\quad
    r=q\, ,\label{KerrBR_transf_to_gen}
\end{align}
{with $\Omega^2, Q$ given by \eqref{Omega}, and} ${P(p)=(a^2-p^2)\tilde P(p)}$.

The corresponding CKY tensors for the Kerr--BR black holes take the following form:
\begin{align}
    h&=\frac{1}{\Omega^3}\left( -r\,\dd r\wedge[\dd t-a\sin^2\theta \dd\varphi]
    \right.
    \nonumber\\
    &\qquad\quad\left.- a\cos\theta\sin\theta\, \dd\theta\wedge[a\dd t-(r^2+a^2)\dd \varphi]\right),\label{eq: h BL}
    \\
    k&=\frac{1}{\Omega^3}\left( -a\cos\theta\, \dd r\wedge[\dd t-a\sin^2\theta\, \dd\varphi]
    \right.
    \nonumber\\
    &\qquad\qquad\left.r\sin\theta\, \dd\theta\wedge[a\dd t-(r^2+a^2)\dd\varphi]\right).\label{eq: k BL}
\end{align}
In fact, up to the conformal factor these are the usual expressions for the Kerr--Newman black hole in Boyer--Lindquist coordinates. 
Note that, as discussed above, the hidden symmetries do not --- in an obvious way --- generate the explicit ones $\partial_t$ and $\partial_\varphi$.

\subsubsection{Separability}
All the separability results derived above carry over to this particular case. 
As these are scalar equations, one can easily apply the coordinate transformations \eqref{KerrBR_transf_to_gen}, and the corresponding relabelling of the parameters
\begin{equation}
    \omega\to \tilde \omega \,, \quad n\to -a^2 \tilde \omega-a \tilde n\,,
\end{equation}
to obtain the Teukolsky-like equations (\ref{pol_teuk_gen})--(\ref{rad_teuk_gen}) for the Kerr--BR black hole.

We include the explicit results for null geodesics here because the above transformation is a little bit more involved. 
Proceeding as before, with the additive separation of variables ansatz (\ref{separation-ansatz}),
\begin{equation}
  S=W_t(t)+W_r(r)+W_\theta(\theta)+W_\varphi(\varphi)\,,
\end{equation}
the two Killing constants of motion give
\begin{equation}
     W_t=-\tilde{E} \,t\,,\qquad W_\varphi= \tilde{L}\,\varphi\,.
\end{equation}
The contravariant conformal Killing tensor reads
\begin{align}
    K^{\mbox{\tiny (KBR)}}_{(h)}&=r^2\bigg[\frac{1}{\rho^2Q}\big((r^2+a^2)\,\partial_t+a\,\partial_\phi\big)^2-\frac{Q}{\rho^2}\partial_r^2\bigg]
    \nonumber\\
    &\hspace{-9mm}+a^2\cos^2\theta\bigg[\frac{1}{\rho^2\tilde P\sin^2\theta}\big(a\sin^2\theta \,\partial_t+\partial_\varphi\big)^2+\frac{\tilde P}{\rho^2}\partial_\theta^2\bigg].
\end{align}
Carter's constant \eqref{Carter1} thus yields the remaining terms
\begin{align}
   \tilde P\,(\dot{W}_\theta)^2&=\tilde C-\dfrac{\big(\tilde{E}\, a^2\sin^2\theta-\tilde {L}\big)^2}{\tilde P\,\sin^2\theta}\,,\label{eq: W_theta}
  \\
  Q\,(W'_r)^2&=\dfrac{\big(\tilde{E}\, (r^2+a^2)-a\tilde{L}\big)^2}{Q}-\tilde{C}\,.
  \label{eq: W_r}
\end{align}
Consequently, using the expression \eqref{eq: velocity momentum HJ relation} we obtain
\begin{align}
    \dfrac{\rho^2}{\Omega^2}\dfrac{\dd t}{\dd\lambda}&=\dfrac{r^2+a^2}{Q}[\tilde{E}(r^2+a^2)-a\tilde{L}]+\dfrac{a}{\tilde{P}}(\tilde{L}-a \tilde{E}\sin^2\theta)\,,\nonumber\\
    \dfrac{\rho^2}{\Omega^2}\dfrac{\dd\varphi}{\dd\lambda}&=\dfrac{a}{Q}[\tilde{E}(r^2+a^2)-a\tilde{L}]+\dfrac{\tilde{L}}{\tilde{P}\sin^2\theta}-\dfrac{a\tilde{E}}{\tilde{P}}\,,\\
    \dfrac{\rho^2}{\Omega^2}\dfrac{\dd r}{\dd\lambda}&=\pm \sqrt{[\tilde{E}(r^2+a^2)-a\tilde{L}]^2-Q \tilde{C}}\,,\nonumber\\
    \dfrac{\rho^2}{\Omega^2}\dfrac{\dd\theta}{\dd\lambda}&=\pm\dfrac{1}{\sin\theta}\sqrt{\tilde{C}\tilde{P}\sin^2\theta-(\tilde{L}-a\tilde{E}\sin^2\theta)^2}\,.\nonumber
\end{align}
These results are equivalent to the recent calculations in e.g. \cite{Wang:2025vsx, Zeng:2025tji, Ali:2025beh, Vachher:2025jsq}.

\subsection{Schwarzschild--BR solution}\label{sec_Schw_BR}

In the non-twisting case the Kerr rotation parameter~$a$ vanishes, and the metric 
(\ref{metr_Kerr--BR}) simplifies to
\be
    \dd s^2=\frac{1}{\Omega^2}\Bigl[-\mathcal{Q}\,\dd t^2+\frac{\dd r^2}{\mathcal{Q}}+r^2\Big(\frac{\dd\theta^2}{\tilde{P}}+\tilde{P}\sin^2\theta\, \dd\varphi^2 \Big)\Bigr]\,,\label{metr_Schw-BR}
\ee
where 
\begin{align}\label{metric-Schw-BR}
\tilde{P}&=1+B^2m^2 p^2\,, \nonumber\\
{\cal{Q}}&= \,\big(1+B^2r^2\big)\big( 1 - B^2m^2 - 2m/r \big)\,,
     \\
\Omega^2& = 1 + B^2 \big[\,r^2 (1-p^2) + \big( 2m\,r  + B^2m^2 r^2 \big)p^2 \big]\,.\nonumber
\end{align}
The corresponding 4-potential,  obtained as the $a\to 0$ limit of  \eqref{4pot_Kerr-BR}, reads 
\begin{align}
    A=&\dfrac{1}{B}\Big[\cos\gamma\, (\Omega-r\,\Omega_{,r}-1)\,\dd \varphi+\sin \gamma\dfrac{\Omega_{,p}}{r}\,\dd t \Big],
\end{align}
where, as in the case of the Kerr-BR solution, $\gamma$~is the duality rotation parameter (${\gamma=0}$ gives purely magnetic field, whereas ${\gamma=\frac{\pi}{2}}$ gives purely electric field). Note also that in this case, the conicity parameter simply reads 
\be 
C=\frac{1}{1+B^2m^2}\,.
\ee 

Also, the CKY tensors significantly simplify, and are given by 
\be \label{hk_Schwarz}
h=-\frac{r}{\Omega^3} \, \dd r\wedge \dd t\,,\quad k=-\frac{r^3}{\Omega^3} \,\sin\theta\, \dd \theta\wedge \dd \varphi\,. 
\ee
Finally, the corresponding separated equations can be easily obtained from the Kerr--BR results with the parameter $a\to 0$.

\subsection{Penrose charges}
Proceeding as outlined in Sec.~\ref{Sec:HS}, we seek to calculate the Penrose charges \eqref{eq: Penrose Charge}. 
To gain some insight, we first calculate these for the pure Bertotti--Robinson spacetime, obtained from the $m\to0$ limit of the Schwarzschild--BR spacetime, and then turn to the two black hole cases. 

\subsubsection{Bertotti--Robinson spacetime}

To start, let us consider the spacetime  \eqref{metr_Schw-BR}, \eqref{metric-Schw-BR} with ${m=0}$. 
To evaluate the integral \eqref{eq: Penrose Charge}, we use the CKY tensors $h$ and $k$ given by \eqref{hk_Schwarz}. 
Namely, taking the integral over a constant 2-sphere with fixed~$t$ and ${r=r_0}$, we have (starting with $h$): 
\be 
 q_{\mbox{\tiny P}}[h]=\frac{1}{8\pi} \int_{0}^{2\pi C} \dd\varphi\int_0^\pi \dd\theta \,(\star J^h)_{\theta\varphi}\label{qph}\,,
\ee
where ${C=1}$. 
This integral can be evaluated explicitly, and it reads
\ba\label{qh-BR}
q_{\mbox{\tiny P}}[h]&=& \frac{1}{8\pi} \,\frac{B^2}{1+B^2r_0^2} \,\frac{4}{3}\pi r_0^3\nonumber\\
&=&\frac{1}{6}B^2 r_0^3(1-B^2r_0^2)+{\cal O}(B^6)\,,
\ea 
while the analogous integral $q_{\mbox{\tiny P}}[k]$ vanishes for the CKY vector~$k$. 
Obviously, the Penrose charge \eqref{qh-BR} measures, in some sense, the energy stored in the magnetic field determined by~$B$, and it diverges as ${r_0\to \infty}$, reflecting the infinite volume filled with magnetic energy. 
This seems in accordance with the interpretation suggested  in \cite{Hull:2025ivk}, where such charges were introduced, as some kind of ``{\em quasi-local} mass''.

However, this interpretation is a bit more subtle. 
To quadratic order in $B$ only, it agrees with the magnetic self energy, calculated from 
\be 
U_{\mbox {\tiny self}}=-\int T^t{}_t\, \dd V\,,
\ee 
where $T^t{}_t$ simply reads ${T^t{}_t=-B^2/(8\pi)}$, and we consider volume inside the radius $r_0$. 
This yields 
\ba 
U_{\mbox {\tiny self}}&=&\frac{\log(B\,r_0+\sqrt{1+B^2r_0^2}\,)}{2B}-\frac{r_0}{2\sqrt{1+B^2r_0^2}}\,\nonumber\\
&=&{\frac{1}{6}B^2r_0^3\Big(1-\frac{9}{10}B^2r_0^2\Big)}+{\cal O}(B^6)\,.
\ea 
At the same time, it was shown in \cite{Hull:2025ivk} that in Einstein spacetimes, i.e. for spacetimes obeying $G_{\mu\nu}+\Lambda g_{\mu\nu}=0$, the Penrose charges match the Komar charges (up to a constant prefactor) associated to the Killing vector which are generated by their respective CKY tensors. 
In our case, the CKY tensors no longer generate Killing vectors. 
{However, we may consider the Komar mass calculated from  $\xi=\partial_t$ at a given radius $r_0$, given by%
\footnote{One can show that the analogous Komar charge associated with $\partial_\varphi$ vanishes.} 
\be 
q_{\mbox{\tiny K}}[\xi]=-\frac{1}{8\pi} \int_{0}^{2\pi C} \dd\varphi\int_0^\pi \dd\theta \, (\star \dd\xi)_{\theta \varphi}\,,
\ee 
or more explicitly: 
\begin{align}
  q_{\mbox{\tiny K}}[\partial_t]&=  \frac{r_0}{2} - \frac{\text{arcsinh}(B r_0)}{2B \sqrt{1 + B^2 r_0^2}}\nonumber\\
  &=\frac{1}{6}B^2r_0^3 \Big(2-\frac{8}{5}B^2r_0^2\Big)+{\cal O}(B^6)\,.
\end{align}
Up to a factor of 2, this coincides with the above Penrose charge in the leading order of small $B$ expansion. 
}

{To summarize our findings, for the Bertotti--Robinson spacetimes the Penrose charge \eqref{eq: Penrose Charge} associated with the CKY tensor $h$ calculates a some kind of quasi-local energy contained in the magnetic field (whereas it vanishes for the CKY tensor $k$). This energy is related to, but not exactly equal to, those obtained by other methods, the Komar integration for example.  }

We shall now proceed to calculate such a charge for the Schwarzschild--BR and Kerr--BR spacetimes. 
The analogous calculations for the Kerr-NUT and C-metric spacetimes are presented in Appendix~\ref{App:onshell}.\footnote{Interestingly, as shown in Appendix~\ref{App:onshell}, in the presence of the NUT parameter, the Penrose charge derived from the CKY tensor $k$, being proportional to the NUT parameter,  no longer vanishes. 
At the same time, the associated with  $k$ Killing vector $\partial_\varphi$ yields an angular momentum type quantity.}

\subsubsection{Schwarzschild--BR}

Let us now return to the full spacetime \eqref{metr_Schw-BR} and its CKY tensors \eqref{hk_Schwarz}. 
As before, the charge associated to the CKY tensor $k$ vanishes, while the one related to $h$ is given by formula \eqref{qph}. 
Interestingly, even for non-trivial $m$ one can find an exact expression for this charge. However, such an expression is quite large and not illuminating. For this reason, we focus on the following two special cases. 
\vspace{2mm}

i) We evaluate it {\em perturbatively in~$B$} for a general ${r=r_0}$. This yields
\be
   q_{\mbox{\tiny P}}[h]=m+\frac{1}{6} B^2 \big(r_0^3-2m\,r_0^2-6 m^3\big)
    +{\cal O}(B^4)\,. \label{Penrose-Schw-Br}
\ee
Of course, this agrees with \eqref{qh-BR} in the limit ${m\to0}$.
\vspace{2mm}

ii) We evaluate it for a general value of~$B$, but only {\em on the horizon},  that is for ${r=r_+}$ given by the root of ${\Delta(r_+)=0}$, namely  ${r_+=2m/(1-B^2m^2)}$. This simply yields
\begin{equation}
     q_{\mbox{\tiny P}}[h]\big|_{r_0=r_+} = C\,m = \frac{m}{1+B^2m^2}\,.
\end{equation}
Interestingly, this agrees with the ``mass'' of the Schwarzschild--BR spacetime suggested in \cite{Podolsky:2025tle}. In such a case, the Penrose charge evaluated at the horizon would give a measure of the mass of the black hole.

\subsubsection{Kerr--BR}

Finally, let us turn to the Kerr--BR spacetime \eqref{metr_Kerr--BR} and its CKY tensors \eqref{eq: h BL} and \eqref{eq: k BL}. As before, the Penrose charge \eqref{eq: Penrose Charge} vanishes for the CKY tensor~$k$. 
We can evaluate the integral \eqref{qph} arising from the CKY $h$ in this spacetime, but only perturbatively for small values of $B$. 
Namely, taking the integral over a 2-sphere with fixed~$t$ and ${r=r_0}$, we explicitly obtain
\begin{align}
  q_{\mbox{\tiny P}}[h] 
   &=m+B^2 
   \Big[\,\frac{1}{6} a^2 (r_0+2m) \nonumber\\  
   &+\frac{1}{6}(r_0^3-6 m r_0^2+5 m^2 r_0-6 m^3)
   \nonumber\\
   & -\frac{1}{4} m (4 r_0-5 m)\frac{r_0^4}{a^3}\Big(\frac{a}{r_0}-\Big(1+\frac{a^2}{r_0^2}\Big)  \arctan\Big(\frac{a}{r_0}\Big)\Big)\Big]
\nonumber\\
   & +{\cal O}(B^4)\,. \label{Penrose-Kerr-Br}
\end{align}
This generalizes \eqref{Penrose-Schw-Br} to include the Kerr rotation parameter~$a$, and of course reduces to it in the limit $a\to0$.

\section{Conclusions}\label{Sec:Conclusions}

Motivated by the newly discovered class of type D spacetimes with non-aligned electromagnetic field \cite{Podolsky:2025tle, Ovcharenko:2025cpm}, in this paper, we have investigated the hidden symmetries and separability properties of the general off-shell conformal-to-Carter spacetimes \eqref{off-shell-PD}, 
\be \label{g_concl}
g=\frac{1}{\Omega^2}\,g^\C\,,
\ee 
where $\Omega$ is an {\em arbitrary conformal factor} and $g^\C$ is the off-shell Carter's metric \eqref{Carter}, with any ${Q=Q(q)}$ and ${P=P(p)}$. 
In particular, we have focused on spacetimes for which the conformal rescaling preserves the isometries/explicit symmetries, i.e. ${\Omega=\Omega(p,q)}$.
The corresponding hidden symmetries are of a {\em weaker} nature than for the Carter spacetime, as they posses only a {\em conformal} Killing--Yano  tensor $h$ which is no longer closed.
Unsurprisingly, we then find only conformal (not full) Killing tensors, from the square of $h$, and the separability properties carry through only for conformal/massless test field equations.  
The situation is similar to the PD case where these weaker objects were already known, e.g.~\cite{Kubiznak:2007kh}.

However, our results go far beyond the previously discussed case of  PD metrics and cover, in particular, the newly discovered OP spacetimes. 
Interestingly, we have shown that for such spacetimes, and contrary to the Carter and PD case, the hidden symmetry no longer generates explicit symmetries in any straightforward way.

Among the conformal-to-Carter spacetimes, we have also discovered a new geometry for which the conformal factor takes a specific form \eqref{novel}, namely
\be 
\Omega=\frac{pq}{1-pq}\,,
\ee 
and for which the isometries can be inferred. 
It remains to be seen whether this new {\em geometrically preferred} metric remains to solve the Einstein equations with some physically well-motivated matter fields.

We have also explicitly demonstrated {\em separability} of various test field or particle equations with intrinsic conformal invariance (massless/conformally coupled) in spacetimes \eqref{g_concl},
 showing that there is a Teukolsky-type equation combining the spin 0, 1/2, and 2 fields. 
 This could prove useful for investigating the physical properties of the OP spacetimes, including (but not limited to) studying the shadows of, and gravitational waves on, the Kerr--BR black holes. 
 Upcoming work will show that the equations of motion for massless spinning particles separate in these spacetimes as a special case of~\cite{Gray2025}.

Finally, we have shown that the OP spacetimes provide a unique testing-ground for studying the recently defined Penrose charges \cite{Hull:2025ivk} whose existence is intrinsically related to hidden rather than explicit symmetries. 
More specifically, we have calculated such charges for the Bertotti--Robinson, Schwarzchild--BR, and Kerr--BR spacetimes in the expansion of small external magnetic field~$B$, observing that such (quasi-local) charges include the information about the magnetostatic energy. 
As shown in the appendix, such charges are also able to detect the presence of NUT charges and acceleration. 
The precise interpretation of these results, however, remains an important open question for future studies.

In closing, let us mention a few more possible future directions. 
From a mathematical perspective, it would be interesting to try to fully characterize the conformal symmetry operators for the massless test fields in a unified language,  working directly with the fields themselves.
Work towards this has been done in the Geroch--Held--Penrose (extension to the Newman--Penrose) formalism~\cite{Araneda:2018ezs} but it is open how to do this when working directly with the fields. 
Perhaps this can be done using the recent characterization of symmetry operators~\cite{Mei:2023pho, Gray:2024kad,Cynthia}.
Another direction is to look for degenerate limits of OP metrics in the sense of  \cite{Frolov:2017whj}, and to explore whether such spacetimes can be cast in some generalized Kerr--Schild form \cite{Hassaine:2024mfs, Srinivasan:2025hro}.

\subsection*{Acknowledgements}

D.K. is grateful for the support from GA{\v C}R 23-07457S grant of the Czech Science Foundation, and the Charles University Research Center Grant No. UNCE24/SCI/016. H.O. and J.P were supported by the Czech Science Foundation Grant No.~GA\v{C}R 23-05914S, and by the Charles University Grant No.~GAUK 260325.

\appendix

\section{On shell spacetimes}
\label{App:onshell}
In this appendix we gather the on-shell metric functions (that is, exact solutions to the field equations) for various kinds of spacetimes studied in the main text.

\subsection{Carter spacetime}\label{sec_Carter}
In order the Carter metric \eqref{Carter} together with the vector potential \eqref{ACarter} satisfy the Einstein--Maxwell equations, the metric functions $Q$ and $P$ have to take the following specific form:
\ba \label{eq: Cart met funcs}
Q&=&(q^2+a^2)\Bigl(1-\frac{\Lambda}{3}q^2\Bigr)-2m\,q+{\rm e}^2+{\rm g}^2\,,\nonumber\\
P&=&(a^2-p^2)\Bigl(1+\frac{\Lambda}{3}p^2\Bigr)+2{\rm n}\,p\,.
\ea 
Such a solution is thus characterized by the following 3 parameters $\{a, m, {\rm n}\}$, related to rotation, mass, and NUT charge, the cosmological constant $\Lambda$, and the electric charge~${\rm e}$ and the magnetic charge~${\rm g}$. 

The Penrose charges can be calculated for both $h$ and~$k$. We present the calculation here only for ${\Lambda=0}$ for compactness. 
This is best done in Boyer--Lindquist coordinates in which the time and angular coordinates have the appropriate asymptotic properties. 
First, we make the coordinate transformation
\begin{align}
    t&=\eta-(a+{\rm n})^2\,\sigma\,,\qquad
    \varphi=-a\,\sigma\,,\nonumber\\
    a& \cos\theta+{\rm n}=p\,,\qquad r=q\,,
\end{align}
which brings $h$ and $k$ to the form:
\begin{align}
    h=&\,r\,\dd r\wedge(\dd t+2{\rm n}(1+\cos\theta)-a\sin^2\theta)\,\dd\phi
    \nonumber\\
    &+({\rm n}+a\cos\theta)\,\sin\theta \dd\theta\wedge(a\dd t -[r^2+(a-{\rm n})^2]\,\dd\phi)\,,
    \\
    k=&\,({\rm n}+a\cos\theta)\dd r\wedge(\dd t+2{\rm n}(1+\cos\theta)-a\sin^2\theta)\,\dd\phi
    \nonumber\\
    &-r\sin\theta \dd\,\theta\wedge(a\dd t -[r^2+(a-{\rm n})^2]\,\dd\phi)\,.
\end{align}
Upon inserting these into the definition \eqref{eq: Penrose Charge}, the charges take the form
\begin{align}\label{eq: Penrose KAdS}
    q_{\mbox{\tiny P}}[h]&=\frac{1}{8\pi} \int_{0}^{2\pi C} \dd\varphi\int_0^\pi \dd\theta \,(\star J^h)_{\theta\varphi}
    \nonumber\\[2mm]
    &=C\,\frac{\big(r_0^2+a^2-{\rm n}^2\big) \big(m-\tfrac{1}{4 r_0}({\rm e}^2+{\rm g}^2)\big)+2 {\rm n}^2 r_0}{r_0^2+(a+{\rm n})^2}
    \nonumber\\
    &\quad +C\,({\rm e}^2+{\rm g}^2)\,\frac{r_0^2+(a-{\rm n})^2}{8 a r_0^2}
    \nonumber\\
    &\quad \times\left[\arctan\left(\frac{r_0}{a-{\rm n}}\right)+\arctan\left(\frac{r_0}{a+{\rm n}}\right)\right],
    \nonumber \\[2mm]
   q_{\mbox{\tiny P}}[k]&= -C\,{\rm n}\,\frac{\tfrac{1}{2}({\rm e}^2+{\rm g}^2)+ \Delta (r_0)}{r_0^2+(a+{\rm n})^2}\,,
\end{align}
where 
\begin{align}\label{eq: NUT Delta}
\Delta(r)=r^2+a^2-{\rm n}^2-2m\,r+{\rm e}^2+{\rm g}^2\,,
\end{align}
$r_0$ is the Boyer--Lindquist radius of the sphere $t,r$~constant over which the integral is taken, and $C$ is the conical parameter.

If the charges vanish, ${{\rm e}=0={\rm g}}$, this simplifies to
\begin{align}\label{eq: Penrose KAdS no charge}
    q_{\mbox{\tiny P}}[h]
    &=C\,\frac{[\Delta(r_0) +2m\,r_0]\,m+2 {\rm n}^2 r_0}{r_0^2+(a+{\rm n})^2}\,,
    \nonumber \\[2mm]
   q_{\mbox{\tiny P}}[k]&= -C\,{\rm n}\,\frac{\Delta (r_0)}{r_0^2+(a+{\rm n})^2}\,.
\end{align}
When also the NUT parameter vanishes, ${{\rm n}=0}$, we obtain ${q_{\mbox{\tiny P}}[h]=C\,m}$ and ${q_{\mbox{\tiny P}}[k]=0}$, which is the expected result for the Kerr black hole. Interestingly, {\em on the black hole horizon} (determined by the condition ${\Delta(r_+)=0}$) the expressions \eqref{eq: Penrose KAdS no charge} reduce to 
\begin{align}
    q_{\mbox{\tiny P}}[h]
    &=2C\,\frac{(m^2+{\rm n}^2)\,r_+}{r_+^2+(a+{\rm n})^2}\,,
    \nonumber \\[2mm]
    q_{\mbox{\tiny P}}[k]&= 0\,.
\end{align}
Moreover,   
using the expression for the mass in terms of the radius of the horizon $r_+$, $m=(r_+^2+a^2-{\rm n}^2)/(2r_+)$, which comes from $\Delta|_{{\rm e}=0={\rm g}}(r_+)=0$, the last formula can also be rewritten as 
\begin{align}
    q_{\mbox{\tiny P}}[h]
         =C\,m + C\,\frac{{\rm n}\,({\rm n}-a)}{r_+}\,.
\end{align}
Note that in here, the parameters $m$ and $r_+$ are no longer independent quantitities.

Finally, by expanding \eqref{eq: Penrose KAdS} in the powers of ${r_0^{-1}}$, we obtain nice explicit formulas
\begin{align}
   q_{\mbox{\tiny P}}[h]&=\ C\,m+{\cal O}(r_0^{-1})\,,\nonumber
    \\
   q_{\mbox{\tiny P}}[k]&=-C\,{\rm n}+{\cal O}(r_0^{-1})\,,
\end{align}

There are some subtleties around this deficit as it cannot be regular on both axes at the same time, see e.g. \cite{BallonBordo:2019vrn}. 
To do this carefully one needs to account for these singular tubes~\cite{Bordo:2019tyh} in the Penrose charges, see also the recent discussion in \cite{Kolar:2025kle}.

\subsection{Pleba\'nski--Demia\'nski spacetime}\label{sec: PD}
{The Pleba\'nski--Demia\'nski spacetime \eqref{PD} is on-shell, provided 
\ba \label{eq: PB met funcs}
Q&=&\kappa+{\rm e}^2+{\rm g}^2-2m\,q+\epsilon q^2-2{\rm n}\,q^3-(\kappa+\Lambda/3)q^4\,,\nonumber\\
P&=&\kappa+2{\rm n}\,p-\epsilon p^2+2m\,p^3-(\kappa+{\rm e}^2+{\rm g}^2+\Lambda/3)p^4\,.\nonumber\\
\ea

Here, again $\Lambda$ stands for the cosmological constant, ${\rm e}$~is the electric charge parameter, and  ${\rm g}$~is the magnetic charge. The other 4 parameters $\{\kappa, \epsilon, m, {\rm n}\}$ are related to mass, rotation, the NUT charge, and acceleration; see  \cite{Griffiths:2005qp, Ovcharenko:2024yyu, Ovcharenko:2025fxg} for details.
One can also calculate the Penrose charges for the PD spacetime, i.e. by following the same procedure as for the Carter spacetime, i.e. transforming to Boyer--Lindquist coordinates.  
In contrast to the standard Komar charges these are quasi-local quantities and so can be evaluated at a finite radius without needing to worry about the asymptotics.
However the resulting expressions are not very enlightening and so we do present them here.

On the other hand, in the subcase of accelerating charged black holes without the NUT parameter (i.e. the $C$-metric with a cosmological constant~$\Lambda$), using the coordinate transformations and parameter redefinitions as in Section 3 and 4.2 of \cite{Griffiths:2005qp}, we find the relatively simple expressions
\begin{align}\label{eq: Penrose CMet}
    q_{\mbox{\tiny P}}[h]&=C\,m
    \nonumber\\
    &+C\,({\rm e}^2+{\rm g}^2)\,\frac{(a^2+r_0^2) \arctan(r_0/a)-a\, r_0}{4 a\, r_0^2}\,,
     \nonumber\\
     q_{\mbox{\tiny P}}[k]&=-\alpha \,C\,ma
     \nonumber\\
     &+\alpha\,C\,({\rm e}^2+{\rm g}^2) \,\frac{(a^2+r_0^2) \arctan(r_0/a)+a\, r_0}{4 a^2}\,.
\end{align}
Here $C$ is again the conicity parameter, the parameter $a$ denotes Kerr-like rotation, and $\alpha$ the acceleration.
These integrals are also done over slices of constant $t,r$. 
Due to the subtleties concerning the regularization of the axis encoded in~$C$ and the asymptotics, it is unclear exactly how to interpret these general expressions.

Notice however, that in the {\em vacuum case} ${{\rm e}=0={\rm g}}$,
\eqref{eq: Penrose CMet} reduces to
\begin{align}\label{eq: Penrose CMet-vacuum}
q_{\mbox{\tiny P}}[h]=C\,m\,,\qquad
q_{\mbox{\tiny P}}[k]=-\alpha \,C\,ma\,,
\end{align}
{\em independently} of
 the choice of~$r_0$. Moreover, the CKY~$h$ generates the Killing vector ${\xi=\partial_t}$, while the CKY~$k$ generates the Killing vector ${\zeta=\alpha\,\partial_\varphi}$, cf.~Table~\ref{Tab: conform factors}. 
Notice also  that ${q_{\mbox{\tiny P}}[k]=0}$ when ${\alpha=0}$. 
This is because then the CKY $h$ becomes closed, and therefore $k$ becomes a KY tensor. 
In this case the Killing vector it generates vanishes and thence there is no associated secondary charge cf. \eqref{Current1form}. 
Thus, it once more seems reasonable to associate  $q_{\mbox{\tiny P}}[h]$ and  $q_{\mbox{\tiny P}}[k]$ with a `quasi-local mass' and `angular momentum', respectively, although the Komar-based equivalent asymptotic concepts are not well defined. 
A proper analysis would require to compare these expressions with thermodynamic quantities, which are well understood in the asymptotically AdS case~\cite{Anabalon:2018qfv}.

\subsection{Ovcharenko--Podolsk\'y spacetime}
The Ovcharenko--Podolsk\'y spacetime \eqref{OP} is on-shell (necessarily with ${\Lambda=0}$), provided the metric functions $P$ and $Q$ and the conformal factor $\Omega$ take the following form \cite{Ovcharenko:2025cpm}:
\begin{align}
 P=&\ \hat{a}_0+\hat{a}_1\,p+\hat{a}_2\,p^2+\hat{a}_3\,p^3+\hat{a}_4\,p^4\,, \nonumber\\
 Q=&\ \hat{b}_0+\hat{b}_1\,q+\hat{b}_2\,q^2+\hat{b}_3\,q^3+\hat{b}_4\,q^4\,,
\nonumber\\
\Omega^2=&\ 1 - 2\,p\,q + c_{10}\,p +c_{01}\,q +c_{02}\,(q^2-p^2) \nonumber\\
    &+c_{12}\,p\,q^2 +c_{21}\,p^2q+c_{22}\,p^2q^2\,,\label{omega_exp_1}
\end{align}
with the coefficients restricted as follows: 
\begin{align}
    \hat{a}_0&= \dfrac{c_{02}-\frac{1}{4}c_{01}^2}{4c'\bar{c}'}\,, \nonumber\\
    \hat{a}_1&= \dfrac{c_{01}+c_{02}c_{10}+c_{12}}{4c'\bar{c}'}\,, \nonumber\\
    \hat{a}_2&=-\dfrac{1+c_{02}^2-c_{10}c_{12}+\frac{1}{2}c_{01}c_{21}-c_{22}}{4c'\bar{c}'}\,,\label{a-hat-i}\\
    \hat{a}_3&= \dfrac{c_{21}+c_{22}c_{10}-c_{02}c_{12}}{4c'\bar{c}'}\,, \nonumber\\
    \hat{a}_4&=-\dfrac{c_{02}c_{22}+\frac{1}{4}c_{21}^2}{4c'\bar{c}'}\,, \nonumber
\end{align}
and
\begin{align}
    \hat{b}_0&= \dfrac{c_{02}+\frac{1}{4}c_{10}^2}{4c'\bar{c}'}\,, \nonumber\\
    \hat{b}_1&=-\dfrac{c_{10}-c_{02}c_{01}+c_{21}}{4c'\bar{c}'}\,, \nonumber\\
    \hat{b}_2&= \dfrac{1+c_{02}^2+\frac{1}{2}c_{10}c_{12}-c_{01}c_{21}-c_{22}}{4c'\bar{c}'}\,,\label{b-hat-i}\\
    \hat{b}_3&=-\dfrac{c_{12}+c_{22}c_{01}+c_{02}c_{21}}{4c'\bar{c}'}\,, \nonumber\\
    \hat{b}_4&= -\dfrac{c_{02}c_{22}-\frac{1}{4}c_{12}^2}{4c'\bar{c}'}\,. \nonumber
\end{align}
The electromagnetic 4-potential, corresponding to this exact spacetime, is given in \eqref{A_OP}. 
The family of solutions is thus characterized by 6 real parameters $\{ c_{10}, c_{01},c_{02},c_{12}, c_{21}, c_{22}\}$ and one complex parameter $c'$. In this text, we have set ${1/{\bar c'}=2({\rm e}+\im\, {\rm g})}$ to work with real parameters, that is ${1/(4c'{\bar c'})={\rm e}^2+{\rm g}^2}$. 
Note, however, that in this case the constants ${\rm e}$ and ${\rm g}$ are no longer straightforwardly related to electric and magnetic charges.

\section{Coordinate transformations}\label{App:B}

In this appendix, we discuss several coordinate transformations used to simplify the off-shell metrics in the main text. 

\subsection{Relation to the Pleba\'nski--Demia\'nski spacetime}

\begin{statement}
Starting from the conformal-to-Carter spacetime \eqref{off-shell-PD} with $\Omega$ given by
\begin{align}
    \Omega=a+c\, p\,q\,,\qquad a,c\neq 0\,,
\end{align}
there exists a coordinate transformation that brings this metric to the off-shell Pleba\'nski--Demia\'nski form \eqref{PD} with ${\Omega=1-p\,q}$\,.
\end{statement}

\begin{proof}
    To start with, let us introduce the following new coordinates $q'$ and $p'$:%
    \footnote{Here we assume that $a,c> 0$. 
    The sign of $a$ is irrelevant and can be assumed to be positive as the metric only depends on  $\Omega^2$. 
    At the same time, $c<0$ case may be obtained by changing $p'\to -p'$ in all the expressions below.}
    \begin{align}
        q=\sqrt{\dfrac{a}{c}}\,q'\,,\qquad
        p=-\sqrt{\dfrac{a}{c}}\,p'\,.
    \end{align}
    After this redefinition, the conformal factor and the combination ${q^2+p^2}$ become 
    \begin{align}
        \Omega=a\,\Omega'\,,\quad
        \Omega'=1-p\,q\,,\quad q^2+p^2=\dfrac{a}{c}(q'^2+p'^2)\,.
    \end{align}
    Now, let us consider each term in the metric separately. 
    That is, let us start with the term $\dd q^2$. If we define $Q'=c^2Q$, then this term becomes
    \begin{align}
        \dfrac{1}{\Omega'^2}\dfrac{q'^2+p'^2}{Q'}\,\dd q'^2\,.
    \end{align}
    The analogous situation is with the $\dd p^2$ term that becomes 
    \begin{align}
        \dfrac{1}{\Omega'^2}\dfrac{q'^2+p'^2}{P'}\,\dd p'^2\,,
    \end{align}
    if we define ${P'=c^2 P}$. 
    Now, let us turn to  the term ${\dd\eta-p^2 \dd\sigma}$. After these definitions and by rescaling the coordinates $\eta$ and $\sigma$ in the following way:
    \begin{align}
        \eta=\sqrt{a^3\, c}\,\eta'\,,\qquad\sigma=\sqrt{ac^3}\,\sigma'\,,
    \end{align}
    this term becomes
    \begin{align}
        -\dfrac{1}{\Omega'^2}\,\dfrac{Q'}{q'^2+p'^2}\,(\dd\eta'-p'^2 \dd\sigma')^2\,.
    \end{align}
    Moreover, the last term becomes
    \begin{align}
        -\dfrac{1}{\Omega'^2}\,\dfrac{P'}{q'^2+p'^2}\,(d\eta'+q'^2 d\sigma')^2\,.
    \end{align}
    Thus, the metric indeed takes the form \eqref{PD} with 
    \begin{align}
        \Omega'=1-p'q'\,,
    \end{align}
    as was required to show.
    \end{proof}

\subsection{Carter spacetime in conformal coordinates}

\begin{statement}
    Starting from the conformal-to-Carter metric \eqref{off-shell-PD} with $\Omega$ given by
\begin{align}
    \Omega = p\,q\,,
\end{align}
there exists a coordinate transformation that brings this metric to the off-shell Carter form \eqref{Carter}.
\end{statement}

\begin{proof}
    The proof of this statement is quite simple. Let us consider the following  transformation of coordinates: 
    \begin{align}
        q=\dfrac{1}{q'}\,,\qquad p=\dfrac{1}{p'}\,,
    \end{align}
    upon which the conformal factor and the combination $q^2+p^2$ become 
    \begin{align}
        \Omega=\dfrac{1}{p'q'}\,,\qquad q^2+p^2=\dfrac{q'^2+p'^2}{p'^2q'^2}\,.
    \end{align}
    The term $\dd q^2$, if we define ${Q'=Q\,q'^4}$, becomes
    \begin{align}
        \dfrac{q'^2+p'^2}{Q'}\,\dd q'^2\,.
    \end{align}
    The term $\dd p^2$, if we define ${P'=P\, p'^4}$, becomes
    \begin{align}
        \dfrac{q'^2+p'^2}{P'}\,\dd p'^2\,.
    \end{align}
    Let us next  consider the term  $\dd\eta-p^2 \dd\sigma$. 
    If we define 
    \begin{align}
        \eta=\sigma'\,,\qquad\sigma=\eta'\,,
    \end{align}
    it becomes
    \begin{align}
        -\dfrac{Q'}{q'^2+p'^2}\,({\dd\eta'-p'^2 \dd\sigma'})^2\,.
    \end{align}
    The term ${\dd\eta+q^2 \dd\sigma}$, analogously, becomes
    \begin{align}
        \dfrac{P'}{q'^2+p'^2}\,(\dd\eta'+q'^2\dd\sigma')^2\,.
    \end{align}
    Thus, the metric takes the form \eqref{Carter}, as we wanted to show.
\end{proof}

\subsection{Canonical form of the novel geometry N}

\begin{statement}
Starting from the conformal-to-Carter metric \eqref{off-shell-PD} with  $\Omega$ given by
\begin{align}
    \Omega=\dfrac{p\, q}{a+c\,p\,q}\,,\qquad a,c\neq 0\,,
\end{align}
there exists a coordinate transformation that brings it to the canonical form
\begin{align}
    \Omega=\dfrac{p\,q}{1-p\,q}\,,
\end{align}
while preserving the form of the metric \eqref{off-shell-PD}.
\end{statement}

\begin{proof}
    To show this, let us start by introducing the new coordinates $q'$ and $p'$ as
    \begin{align}
        q=\sqrt{\dfrac{a}{c}}\,q'\,,\qquad
        p=-\sqrt{\dfrac{a}{c}}\,p',
    \end{align}
    Here, again, we have assumed W.L.O.G. that $a,c> 0$, see above. 
    Then
    \begin{align}
        \Omega=-\dfrac{1}{c}\,\Omega'\,,\quad
        \Omega'=\dfrac{p\,q}{1-p\,q}\,,\quad
        q^2+p^2=\dfrac{a}{c}(q'^2+p'^2)\,.
    \end{align}
    If we define ${Q'=a^2Q}$, the first term $\dd q^2$ becomes
    \begin{align}
        \dfrac{1}{\Omega'^2}\,\dfrac{q'^2+p'^2}{Q'}\,\dd q'^2\,.
    \end{align}
    Similarly, the term $\dd p^2$, upon setting ${P'=a^2 P}$,   becomes 
    \begin{align}
        \dfrac{1}{\Omega'^2}\,\dfrac{q'^2+p'^2}{P'}\,\dd p'^2\,.
    \end{align}
    Considering next the term $\dd\eta-p^2 \dd\sigma$, and rescaling the coordinates $\eta$ and $\sigma$ as follows: 
    \begin{align}
        \eta=\dfrac{\eta'}{\sqrt{a c^3}}\,,~~~\sigma=\dfrac{\sigma'}{\sqrt{a^3 c}}\,,
    \end{align}
    we arrive at 
    \begin{align}
        -\dfrac{1}{\Omega'^2}\,\dfrac{Q'}{q'^2+p'^2}\,(\dd\eta'-p'^2 \dd\sigma')^2\,.
    \end{align}
    Moreover, the last term becomes
    \begin{align}
        \dfrac{1}{\Omega'^2}\,\dfrac{P'}{q'^2+p'^2}\,(\dd\eta'+q'^2 \dd\sigma')^2\,.
    \end{align}
    Thus, the metric indeed takes the form \eqref{PD} with 
    \begin{align}
        \Omega'=\dfrac{p'\,q'}{1-p'\,q'}\,,
    \end{align}
    as we wished to show. 
    \end{proof}

\bibliography{ref.bib}
\bibliographystyle{apsrev4-1}

\end{document}